\documentclass[11pt]{article}

\usepackage[a4paper,margin=1in]{geometry}

\usepackage[utf8]{inputenc}
\usepackage{authblk}
\usepackage[authoryear]{natbib}

\usepackage{amsmath,amssymb,amsthm}

\usepackage{graphicx}
\usepackage[small]{caption}
\usepackage{booktabs}
\usepackage{multirow}

\usepackage{algorithm}
\usepackage[noend]{algpseudocode}

\usepackage[hidelinks]{hyperref}
\usepackage{url}
\usepackage{xspace}
\usepackage{tikz}
\usepackage{tikzscale}
\usepackage{soul}

\newcommand{\NP}{\textsc{NP}}
\newcommand{\TMB}{\textsc{TMB}\xspace}
\newcommand{\G}{\ensuremath{\mathcal{G}}\xspace}
\newcommand{\T}{\ensuremath{\mathcal{T}}\xspace}
\DeclareMathOperator{\D}{\mathcal{D}\xspace}
\DeclareMathOperator{\tr}{tr}
\DeclareMathOperator{\ta}{at}
\DeclareMathOperator{\td}{dt}
\DeclareMathOperator{\ttr}{tt}
\DeclareMathOperator{\EA}{EA}
\DeclareMathOperator{\LD}{LD}
\DeclareMathOperator{\FT}{FT}
\DeclareMathOperator{\ST}{ST}
\DeclareMathOperator{\MH}{MH}
\DeclareMathOperator{\MW}{MW}
\DeclareMathOperator{\dur}{dur}
\DeclareMathOperator{\wait}{wt}
\newcommand{\RF}{\textsc{ReachFast}\xspace}

\newtheorem{theorem}{Theorem}

\newtheorem{lemma}{Lemma}

\newtheorem{definition}{Definition}
\newtheorem{corollary}{Corollary}

\newtheorem{claim}{Claim}

\title{Optimizing Distances for Multi-Broadcast in Temporal Graphs}

\author[1]{Daniele Carnevale\thanks{Work done while the author was a PhD student at the Gran Sasso Science Institute (GSSI).}}
\author[2]{Gianlorenzo D'Angelo}

\affil[1]{University of Geneva\\ \texttt{daniele.carnevale@unige.ch}}
\affil[2]{Gran Sasso Science Institute (GSSI)\\ \texttt{gianlorenzo.dangelo@gssi.it}}

\date{}

\begin{document}
\maketitle

\begin{abstract}
    Temporal graphs represent networks in which connections change over time, with edges available only at specific moments. 
    Motivated by applications in logistics, multi-agent information spreading, and wireless networks, we introduce the \emph{$\D$-Temporal Multi-Broadcast} ($\D$-TMB) problem, which asks for scheduling the availability of edges so that a predetermined subset of sources reach all other vertices while optimizing the worst-case temporal distance $\D$ from any source. We show that $\D$-TMB generalizes REACHFAST \cite{reachfast}.
    We then characterize the computational complexity and approximability of $\D$-TMB under six definitions of temporal distance $\D$, namely Earliest-Arrival ($\EA$), Latest-Departure ($\LD$), Fastest-Time ($\FT$), Shortest-Traveling ($\ST$), Minimum-Hop ($\MH$), and Minimum-Waiting ($\MW$).    
    For a single source, we show that $\D$-\textsc{TMB} can be solved in polynomial time for $\EA$ and $\LD$, while, for the other temporal distances, it is \NP-hard and hard to approximate within a factor that depends on the adopted distance function. We provide a matching approximation algorithm for $\FT$ and $\MW$.
    For multiple sources, if feasibility is not assumed a priori, the problem is inapproximable within any factor, unless $\textsc{P}=\NP$, already for two sources. We complement this negative result by identifying structural conditions that guarantee tractability for $\EA$ and $\LD$ for any number of sources.
\end{abstract}

\section{Introduction}\label{sec:intro}
Temporal graphs provide a natural model for dynamic systems where interactions occur at specific times, such as communication networks, transportation systems, or epidemiological contact networks~\cite{holme2012Temporal,KKK00}.
In a temporal graph, each edge is associated with a set of timestamps indicating the times at which each edge is available. A \emph{temporal path} is a path
that respects the flow of time, i.e., the timestamps of the edges are traversed in non-decreasing order.
The classical notion of shortest-path fails to capture temporal constraints, thus motivating the study of other distance measures for temporal paths. For example, one may seek a temporal path that arrives as early as possible (Earliest-Arrival Path, \textsc{EA}), starts as late as possible (Latest-Departure Path, \textsc{LD}), minimizes the overall duration (Fastest-Time Path, \textsc{FT}), minimizes total traveling time (Shortest-Traveling Path, \textsc{ST}), minimizes number of used edges (Minimum-Hop Path, \textsc{MH}), or minimizes the total waiting time (Minimum-Waiting Path, \textsc{MW}). Note that these temporal measures are not proper distances; for example, they might not satisfy the triangle inequality. 

Motivated by applications in robotics, information spreading, and logistics, considerable research attention has recently been devoted to temporal network optimization problems, where the aim is to schedule the availability of edges in order to optimize some network property, see e.g.~\cite{deligkas2020optimizing,klobas2023interference,deligkas2023being,reachfast,kunz2023in,carnevale2025approximating,deligkas2025how}.
In this work, we focus on the problem of scheduling the availability of the edges of a graph in such a way that the worst-case temporal distance from a set of source vertices to all other vertices is optimal, according to a temporal distance measure $\D$. We introduce the $\D$-\textsc{Temporal Multi-Broadcast} problem, where we are given an underlying static graph, a \emph{traversal} function that associates each edge-time pair with the time needed to traverse it, a \emph{multiplicity} function indicating how many times an edge can appear, and a set of sources, and we must schedule the availability of edges over time in such a way that the multiplicity constraints are satisfied and the worst-case distance $\D\in \{\EA,\LD,\FT,\ST,\MH,\MW\}$ between a source and any other node is optimal. In other words, we assume a fixed infrastructure (a static graph with given traversal times) and an offline planner that assigns time labels to edges, subject to multiplicity constraints, to optimize the worst-case temporal distance from the sources.

The $\D$-\textsc{Temporal Multi-Broadcast} problem has several applications in logistics, multi-agent information spreading, and wireless networks.
Consider a set of different suppliers that distribute different kinds of goods to a set of distribution points or warehouses connected by a road network. They must schedule deliveries over the network in such a way that all different kinds of goods are delivered to all distribution points as quickly as possible, optimizing a chosen temporal distance measure.
We give a concrete example in Figures~\ref{fig:exampleone}--\ref{fig:exampletwo}.
In the network of Figure~\ref{fig:exampleone}, there are two suppliers that distribute two different products (e.g., eggs and milk) to all the distribution points in the network (e.g., the stores of a supermarket chain distributed over a region). An edge represents a road connection that can be used by trucks to deliver goods between two nodes. A pair $(t,w)$ next to an edge $e$ indicates that a truck can be scheduled at time $t$ on edge $e$ and the travel time along edge $e$ at time $t$ is $w$. We assume that no more than one connection can be established on each edge due to the availability of trucks. The distribution points can also be used to store the goods directed to other nodes. 
When a truck travels along an edge $e=(u,v)$ at a time $t$, it first loads the goods that were stored in $u$ before time $t$, and delivers them to $v$ at time $t+w$, where $w$ is the travel time of $e$ at time $t$. 
For example, supplier $M$ can send his goods to $v_3$ via edge $(M,v_3)$ at time $2$; at time 3, the goods reach $v_3$ and can be stored there so that they can also be sent to $v_4$ and $v_2$ at time 9 and 4, respectively.
The goal is to deliver both goods to all distribution points by optimizing the worst-case temporal measure. Several objective functions can naturally be defined. Figure~\ref{fig:exampletwo} (left) represents a schedule that minimizes the time when all the distribution points receive both kinds of goods: node $v_1$ is the last node to receive the goods from supplier $M$ at time 10 via the path passing through $v_3$, $v_2$, and $v_4$. Figure~\ref{fig:exampletwo} (center)  represents a schedule that minimizes the maximum
travel time over all goods: one of the longest travel time is from supplier $M$ to distribution point $v_2$, which takes a total of 3 time units via the path $M,v_3,v_2$. Finally, Figure~\ref{fig:exampletwo} (right) shows a schedule that minimizes the maximum time a good spends on the road: both goods depart from their respective suppliers and arrive at all distribution points, with at most 3 time steps spent in transit.

A similar problem arises in information spreading over a social network, where a set of agents, each holding different data, want to spread their data in such a way that all network users receive all the data as early as possible. Two users can exchange information through a \emph{connection} between them (e.g., an online or in-person meeting) in which users share the data they received so far, see e.g.~\cite{reachfast}. A connection can be scheduled at specific time slots, but a connection between the same two users cannot occur more than a fixed number of times. The aim is to schedule all connections in such a way that the time when the last user receives all the data is minimum. 
Similarly, in a wireless sensor network, the data gathered by a set of sensors must be broadcast to all the nodes as early as possible, but no more than a given number of connections can be established among nodes for energy-saving purposes.

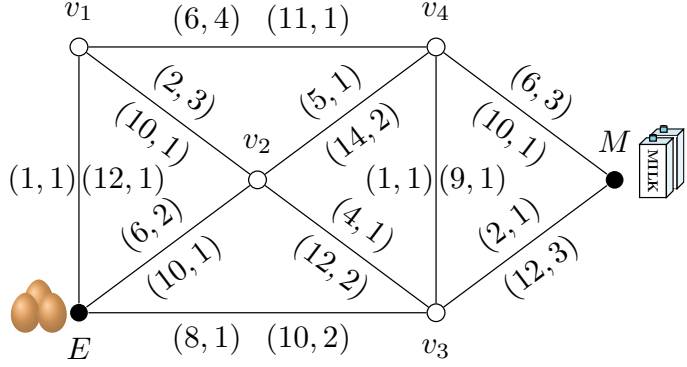
\begin{figure}[t]
\begin{center}
\resizebox{.58\columnwidth}{!}{\small
\begin{tikzpicture}
    \tikzstyle{every node}=[inner sep=2pt,circle,black,fill=black]
    \path (0,-1.5) node (s1)[label=below:$E$]{};
    \path (6,0) node (s2)[label=above:$M$]{};
    \tikzstyle{every node}=[inner sep=2pt,circle,draw,fill=white]
    \path (0,1.5) node (v1)[label=above:$v_1$]{}; 
    \path (2,0) node (v2)[label=above:$v_2$]{};
    \path (4,1.5) node (v4)[label=above:$v_4$]{};
    \path (4,-1.5) node (v3)[label=below:$v_3$]{};
    \tikzstyle{every node}=[]
    
    \draw (s1)-- node[left=-3pt,pos=.5]{$(1,1)$} node[right=-3pt,pos=.5]{$(12,1)$} (v1);
    \draw (s1)-- node[above,pos=.5,sloped]{$(6,2)$} node[below,pos=.5,sloped]{$(10,1)$} (v2);
    \draw (v1)-- node[above,pos=.5,sloped]{$(2,3)$} node[below,pos=.5,sloped]{$(10,1)$} (v2);
    \draw (v2)-- node[above,pos=.5,sloped]{$(4,1)$} node[below,pos=.5,sloped]{$(12,2)$} (v3);
    \draw (v2)-- node[above,pos=.5,sloped]{$(5,1)$} node[below,pos=.5,sloped]{$(14,2)$} (v4);
    \draw (v3)-- node[left=-3pt,pos=.5]{$(1,1)$} node[right=-3pt,pos=.5]{$(9,1)$} (v4);
    \draw (v3)-- node[above,pos=.5,sloped]{$(2,1)$} node[below,pos=.5,sloped]{$(12,3)$} (s2);
    \draw (v4)-- node[above,pos=.5,sloped]{$(6,3)$} node[below,pos=.5,sloped]{$(10,1)$} (s2);
    \draw (v1) -- node[above,pos=.35]{$(6,4)$} node[above,pos=.65]{$(11,1)$} (v4);
    \draw (s1) -- node[below,pos=.35]{$(8,1)$} node[below,pos=.65]{$(10,2)$} (v3);

    \begin{scope}[shift={(-0.45,-1.35)},scale=0.15]
      \def\eggheight{1.5cm}
      \path[preaction={fill=orange!50!white},
      ball color=orange!60!gray,fill opacity=.5]
      plot[domain=-pi:pi,samples=100]
      ({.78*\eggheight *cos(\x/4 r)*sin(\x r)},{-\eggheight*(cos(\x r))})
      -- cycle;
    \end{scope}
    \begin{scope}[shift={(-0.6,-1.5)},scale=0.15]
      \def\eggheight{1.5cm}
      \path[preaction={fill=orange!50!white},
      ball color=orange!60!gray,fill opacity=.5]
      plot[domain=-pi:pi,samples=100]
      ({.78*\eggheight *cos(\x/4 r)*sin(\x r)},{-\eggheight*(cos(\x r))})
      -- cycle;
    \end{scope}
    \begin{scope}[shift={(-0.3,-1.5)},scale=0.15]
        \def\eggheight{1.5cm}
        \path[preaction={fill=orange!50!white},
        ball color=orange!60!gray,fill opacity=.5]
        plot[domain=-pi:pi,samples=100]
        ({.78*\eggheight *cos(\x/4 r)*sin(\x r)},{-\eggheight*(cos(\x r))})
        -- cycle;
    \end{scope}

    \begin{scope}[shift={(6.4,-0.1)},scale=0.15]
        \definecolor{milkblue}{RGB}{173, 216, 230}
        \definecolor{milkwhite}{RGB}{255, 255, 255}
        \definecolor{capcolor}{RGB}{135, 206, 235}

        \draw[fill=milkwhite] (0,0) -- (2,0) -- (2,4) -- (0,4) -- cycle;
        
        \draw[fill=milkblue!30] (2,0) -- (2.5, 0.5) -- (2.5, 4.5) -- (2, 4) -- cycle;
        
        \draw[fill=milkblue!50] (0,4) -- (2, 4) -- (2.5, 4.5) -- (0.5, 4.5) -- cycle;
        \node[scale=0.45,font=\bfseries,rotate=-90] at (1, 2.4) {\ \ MILK};
        \draw[white] (-0.5,-0.14)--(2.5,-0.14);
        \draw[white] (-0.5,4.6)--(2.8,4.6);
        \draw[white] (-0.14,0)--(-.14,4.2);
        \draw[white] (2.6,0)--(2.6,4.6);
        \draw[white] (2.06,-.14)--(2.7,0.5);
        \draw[white] (-0.25,3.94)--(.6,4.8);

        \coordinate (CapLeft) at (0.75,4.25);
        \coordinate (CapRight) at (1.25,4.25);
        \coordinate (CapLeftTop) at (0.75,4.65);
        \coordinate (CapRightTop) at (1.25,4.65);

        \draw[fill=capcolor, draw=black] 
            (CapLeft) -- (CapRight) -- (CapRightTop) -- (CapLeftTop) -- cycle;

        \draw[fill=capcolor!50!black] 
            (0.75,4.7) 
            arc[start angle=180,end angle=0,x radius=.25,y radius=0.05];
    \end{scope}
    \begin{scope}[shift={(6.28,-0.2)},scale=0.15]
        \definecolor{milkblue}{RGB}{173, 216, 230}
        \definecolor{milkwhite}{RGB}{255, 255, 255}
        \definecolor{capcolor}{RGB}{135, 206, 235}

        \draw[fill=milkwhite] (0,0) -- (2,0) -- (2,4) -- (0,4) -- cycle;
        
        \draw[fill=milkblue!30] (2,0) -- (2.5, 0.5) -- (2.5, 4.5) -- (2, 4) -- cycle;
        
        \draw[fill=milkblue!50] (0,4) -- (2, 4) -- (2.5, 4.5) -- (0.5, 4.5) -- cycle;
        \node[scale=0.45,font=\bfseries,rotate=-90] at (1, 2.4) {\ \ MILK};
        \draw[white] (-0.5,-0.14)--(2.5,-0.14);
        \draw[white] (-0.5,4.6)--(2.8,4.6);
        \draw[white] (-0.14,0)--(-.14,4.2);
        \draw[white] (2.64,0)--(2.64,0.7);
        \draw[white] (2.64,4)--(2.64,4.6);
        \draw[white] (2.06,-.14)--(2.7,0.5);
        \draw[white] (-0.25,3.94)--(.6,4.8);

        \coordinate (CapLeft) at (0.75,4.25);
        \coordinate (CapRight) at (1.25,4.25);
        \coordinate (CapLeftTop) at (0.75,4.65);
        \coordinate (CapRightTop) at (1.25,4.65);

        \draw[fill=capcolor, draw=black] 
            (CapLeft) -- (CapRight) -- (CapRightTop) -- (CapLeftTop) -- cycle;

        \draw[fill=capcolor!50!black] 
            (0.75,4.7) 
            arc[start angle=180,end angle=0,x radius=.25,y radius=0.05];
    \end{scope}
\end{tikzpicture}
}
\caption{Instance of $\D$-\textsc{TMB} modeling a road network. Vertices $E$ and $M$ are sources (suppliers); each label $(t,w)$ on an edge indicates an available connection at time $t$ with traversal time $w$. In this example we assume that at most one label can be selected per edge.}
\label{fig:exampleone}
\end{center}
\end{figure}

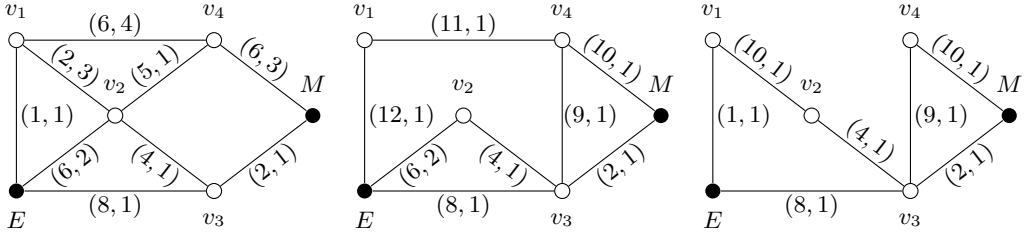
\begin{figure}[t]
\begin{center}
\scalebox{1}{\footnotesize
\begin{tikzpicture}[]
    \tikzstyle{every node}=[inner sep=2pt,circle,black,fill=black]
    \path (0,-1) node (s1)[label=below:$E$]{};
    \path (3.9,0) node (s2)[label=above:$M$]{};
    \tikzstyle{every node}=[inner sep=2pt,circle,draw,fill=white]
    \path (0,1) node (v1)[label=above:$v_1$]{}; 
    \path (1.3,0) node (v2)[label=above:$v_2$]{};
    \path (2.6,1) node (v4)[label=above:$v_4$]{};
    \path (2.6,-1) node (v3)[label=below:$v_3$]{};
    \tikzstyle{every node}=[]
    \draw (s1)-- node[right=-2pt,pos=.5]{$(1,1)$}  (v1);
    \draw (s1)-- node[below=-2pt,pos=.5,sloped]{$(6,2)$} (v2);
    \draw (v1)-- node[above=-2pt,pos=.5,sloped]{$(2,3)$} (v2);
    \draw (v2)-- node[below=-2pt,pos=.5,sloped]{$(4,1)$} (v3);
    \draw (v2)-- node[above=-2pt,pos=.5,sloped]{$(5,1)$} (v4);
    \draw (v3)-- node[below=-2pt,pos=.5,sloped]{$(2,1)$} (s2);
    \draw (v4)-- node[above=-2pt,pos=.4,sloped]{$(6,3)$} (s2);
    \draw (v1)-- node[above=-2pt,pos=.5]{$(6,4)$} (v4);
    \draw (s1)-- node[below=-2pt,pos=.5]{$(8,1)$} (v3);
\end{tikzpicture}
\begin{tikzpicture}[]
    \tikzstyle{every node}=[inner sep=2pt,circle,black,fill=black]
    \path (0,-1) node (s1)[label=below:$E$]{};
    \path (3.9,0) node (s2)[label=above:$M$]{};
    \tikzstyle{every node}=[inner sep=2pt,circle,draw,fill=white]
    \path (0,1) node (v1)[label=above:$v_1$]{}; 
    \path (1.3,0) node (v2)[label=above:$v_2$]{};
    \path (2.6,1) node (v4)[label=above:$v_4$]{};
    \path (2.6,-1) node (v3)[label=below:$v_3$]{};
    \tikzstyle{every node}=[]
    \draw (s1)-- node[right=-2pt,pos=.5]{$(12,1)$}  (v1);
    \draw (s1)-- node[below=-2pt,pos=.5,sloped]{$(8,1)$} (v3);
    \draw (v3)-- node[right=-3pt,pos=.5]{$(9,1)$} (v4);
    
    \draw (s1)-- node[below=-2pt,pos=.5,sloped]{$(6,2)$} (v2);
    
    \draw (v2)-- node[below=-2pt,pos=.5,sloped]{$(4,1)$} (v3);
    \draw (v3)-- node[below=-2pt,pos=.5,sloped]{$(2,1)$} (s2);
    \draw (v4)-- node[above=-2pt,pos=.4,sloped]{$(10,1)$} (s2);
    \draw (v1)-- node[above=-2pt,pos=.5]{$(11,1)$} (v4);
\end{tikzpicture}
\begin{tikzpicture}[]
    \tikzstyle{every node}=[inner sep=2pt,circle,black,fill=black]
    \path (0,-1) node (s1)[label=below:$E$]{};
    \path (3.9,0) node (s2)[label=above:$M$]{};
    \tikzstyle{every node}=[inner sep=2pt,circle,draw,fill=white]
    \path (0,1) node (v1)[label=above:$v_1$]{}; 
    \path (1.3,0) node (v2)[label=above:$v_2$]{};
    \path (2.6,1) node (v4)[label=above:$v_4$]{};
    \path (2.6,-1) node (v3)[label=below:$v_3$]{};
    \tikzstyle{every node}=[]
    \draw (s1)-- node[right=-2pt,pos=.5]{$(1,1)$} (v1);
    \draw (v1)-- node[above=-2pt,pos=.4,sloped]{$(10,1)$} (v2);
    \draw (v2)-- node[above=-2pt,pos=.5,sloped]{$(4,1)$} (v3);
    \draw (s1)-- node[below=-2pt,pos=.5,sloped]{$(8,1)$} (v3);
    \draw (v3)-- node[below=-2pt,pos=.5,sloped]{$(2,1)$} (s2);
    \draw (v4)-- node[above=-2pt,pos=.4,sloped]{$(10,1)$} (s2);
    \draw (v3)-- node[right=-2pt,pos=.5]{$(9,1)$} (v4);
\end{tikzpicture}}
\caption{Three solutions for the $\D$-\textsc{Temporal Multi-Broadcast} instance of Figure~\ref{fig:exampleone}, where the aim is to minimize the maximum earliest arrival time, the maximum duration of travels, and the maximum travel time, respectively.}
\label{fig:exampletwo}
\end{center}
\end{figure}
The $\D$-\textsc{Temporal Multi-Broadcast} problem generalizes the \textsc{ReachFast} problem, introduced in~\cite{reachfast}, 
where we are given a temporal graph and are allowed to shift the temporal edges, i.e., to change the times at which edges appear, in such a way that the latest earliest arrival time from a source to a node is minimum. We prove that, when we consider only earliest arrival distances, problems \textsc{ReachFast} and $\EA$-\textsc{Temporal Multi-Broadcast} are equivalent under value-preserving reductions (see Section~\ref{sec:defproblems}).
 ~\cite{reachfast} proved that \textsc{ReachFast} is tractable for a single source, for two sources when the underlying graph consists of a number of parallel paths between them, and for multiple sources when the underlying graph is a tree. Moreover, they showed that it is \NP-complete to decide whether \textsc{ReachFast} with two sources admits a solution of value 6, which implies that the problem is \textsc{APX}-hard.

\subsection{Our Results}
We begin by introducing formal definitions of $\D$-\textsc{ReachFast} and $\D$-\textsc{Temporal Multi-Broadcast}, and show their equivalence for any $\D$, which allows us to focus exclusively on the latter. We then characterize the computational complexity and approximability of $\D$-\textsc{Temporal Multi-Broadcast} in terms of distance function and number of sources.

In the single-source setting, the problem is tractable for \textsc{EA} (as already established in \cite{reachfast}) and we show that it is also tractable for \textsc{LD}, while for the other temporal distances we prove that it is \NP-complete to decide whether the problem admits a solution of a given value. This latter result implies that the problem cannot be approximated within a given factor that depends on the specific temporal distance measure, unless $\textsc{P}=\NP$. In particular, for $\D\in \{\ST,\MH\}$, $\D$-\textsc{Temporal Multi-Broadcast} cannot be approximated within a factor better than 2, while,
for $\D\in \{\FT,\MW\}$, it cannot be approximated within any factor smaller than an exponential in the number of nodes or a value depending on the input weights and the distance function. In the latter case, we further show that a simple approximation algorithm matches this lower bound.

For the case of multiple sources, we prove that even deciding whether a feasible solution exists is \NP-complete, regardless of the adopted temporal distance function. Consequently, if feasibility is not assumed a priori, the problem cannot be approximated within any finite factor. Moreover, this result holds for any fixed number of sources greater than one. 
We complement these negative results by identifying structural conditions that guarantee tractability for \textsc{EA} and \textsc{LD}. We establish that for \textsc{EA} and \textsc{LD}, the problem admits a polynomial-time algorithm for any number of sources, when the multiplicity is at least equal to the number of sources or when the underlying static graph is a tree and the multiplicity of each edge is at least two.

Our results are summarized in Table~\ref{tab:results}.

\begin{table}[t]
\centering
\caption{Overview of the computational complexity and hardness of approximation for $\D$-\textsc{Temporal Multi-Broadcast} problem, depending on the chosen temporal distance measure $\D$ and the number of sources. The hardness of approximation bounds hold under $\textsc{P}\neq \NP$.}
\begin{tabular}{ccc}
\toprule
\textbf{Distance $\D$} & \textbf{Single source} & \textbf{$\geq 2$ sources} \\
\midrule
EA, LD & Poly-time & \multirow{3}{*}{\parbox{2.4cm}{\centering \NP-complete \& Inapproximable within any factor}} \\
ST, MH & \NP-complete \& \textsc{APX}-hard (factor $<2$) & \\
FT, MW & \NP-complete \& not in \textsc{APX} (factor $<2^{|V|^{O(1)}})$ & \\
& & \\
\midrule
\multirow{2}{*}{EA, LD} & \multicolumn{2}{c}{Poly-time if the multiplicity is at least $|S|$} \\
 & \multicolumn{2}{c}{Poly-time for trees if the multiplicity is at least $2$} \\
\bottomrule
\end{tabular}

\label{tab:results}
\end{table}

\subsection{Related Work}
Delaying (i.e., postponement of edge availability) and shifting (i.e., postponement or advancement of edge availability) are common tools to influence the temporal reachability of a temporal graph. Both have been studied in several works as means to satisfy specific connectivity objectives. Motivated by application in epidemiology, \cite{deligkas2020optimizing} introduced the operation of delaying and studied the problem of minimizing reachability under a bounded number of delays. They proved that minimizing reachability, as well as minimizing the maximum or average reachability, is \NP-hard when a given number of edges can each be delayed by at most a specified input parameter. The problems become tractable when an unbounded number of delays is allowed. 

\cite{kutnersandshift} studied whether, given a temporal graph with passengers (i.e., vertices), each having a specified destination and deadline, it is possible to delay edges so that all passengers reach their targets within the required time. They proved that the problem is \NP-complete both for graphs of lifetime 2 and for planar graphs with larger lifetime. On the positive side, the problem is tractable on trees and fixed-parameter tractable (FPT) when parametrized by the number of demands plus the size of the feedback edge set of the underlying graph. They also provide a comparison of their problem with other problems in the literature involving delaying or shifting (see Table 1 in \cite{kutnersandshift}), which we refer the interested reader to.

Temporal branching and temporal spanning subgraphs~\cite{spanningtreesEA,packingaborescence,branchings} are similar to our problem in the single-source case. A temporal branching is a directed out-tree rooted at a source that spans all vertices, whereas in a temporal spanning subgraph, the underlying graph is not required to be a tree but is still required to be of minimum size. ~\cite{spanningtreesEA} proved that, assuming the source can reach every other vertex, it is possible to compute both a temporal branching from the source and an $\EA$-branching, which is a branching in which the source reaches every other vertex at the earliest possible arrival time. ~\cite{branchings} studied temporal branchings and temporal spanning subgraphs using the same six distance measures adopted in this work. For each distance $\D$, their objective is to find a temporal branching (or a spanning subgraph) rooted at a given vertex that reaches every other vertex via a temporal path whose value under distance $\D$ is minimized. They proved that finding $\D$-branchings and $\D$-spanning subgraphs is \NP-complete for every temporal distance but $\EA$. The main difference between their problem and our single-source case is that we only require the \emph{maximum} distance between the source and any vertex to be minimized.

\section{Preliminaries}\label{sec:pre}
We denote by $\mathbb{N}$ the set of positive integers, and by $\mathbb{N}_0 := \mathbb{N}\cup \{0\}$ the set of nonnegative integers. For an integer $k\in \mathbb{N}$, let $[k]:=\{1,2,\dots,k\}$ and $[k]_0:= \{0\} \cup [k]$.

In this paper, we consider finite loopless graphs. Let $G=(V,E)$ be a simple undirected graph with $n=|V|$ vertices and $m=|E|$ edges. 
For an edge $e=\{u,v\}\in E$, we refer to $u$ and $v$ as the \emph{endpoints} of the edge $e$. Moreover, we say that a vertex $u$ is \emph{incident} to edge $e$ and that $u$ and $v$ are \emph{adjacent}.
A \emph{path} $P$ from a vertex $u$ to a distinct vertex $v$ is a sequence of distinct edges $e_1, e_2, \dots, e_k$ such that $u$ is incident to $e_1$, $v$ is incident to $e_k$, and two edges share an endpoint if and only if they are consecutive in the sequence. 
We say that a path consisting of $k$ edges has \emph{length} $k$.

A \emph{temporal graph} $\G=(G,\lambda,\tr)$ is defined over a finite time span $\tau \in \mathbb{N}$. Here, $G=(V,E)$ is a static graph (called the \emph{underlying graph}), $\lambda:E \to 2^{[\tau]}$ assigns a set of time labels to each edge of $G$, and $\tr : E \times [\tau] \to \mathbb{N}_0$ is the \emph{traversal function}.
A pair $(e,t)$, where $e \in E$ and $t \in \lambda(e)$, is called a \emph{temporal edge}. The traversal function $\tr$ assigns a traversal time to each temporal edge. Specifically, if $\tr(e,t)=x$, this means that traversing edge $e$ from one of its endpoints at time $t$ results in arriving at the other endpoint at time $t+x$.

To \textit{shift} a temporal edge $(e,t)$ means replacing $t$ with $t+x$ in $\lambda$, where $x\in \mathbb{Z}\setminus \{0\}$ is a integer such that $1\leq t+x \leq \tau$. 

A \textit{temporal path} $P$ from $u$ to $v$ is a sequence of temporal edges $(e_1,t_1),(e_2,t_2),\dots,(e_k,t_k)$ such that $e_1,e_2,\dots,e_k$ form a path from $u$ to $v$ in $G$, $t_i\in \lambda(e_i)$ for all $i\in [k]$, and $t_j+\tr(e_j,t_j)\leq t_{j+1}$ for $j\in [k-1]$. If such a path exists, we say that $u$ can \textit{temporally reach} $v$. Let $|P|$ be the number of temporal edges in $P$, i.e., $k$. Moreover, let $\mathcal{P}_{\G}(u,v)$ denote the set of all valid temporal paths from $u$ to $v$ in $\G$.

A temporal spanning out-tree (TSOT) rooted at vertex $s$ is a temporal graph $\G$ whose underlying graph $G$ is a tree, $|\lambda(e)|=1$ for all $e\in G$, and $s$ can temporally reach every other vertex. We denote by $\T_{s}$ such a TSOT rooted at vertex~$s$.

Given a graph $G$ and a traversal function $\tr$, the temporal graph $\G_{\tr}=(G,\lambda,\tr)$, where $\lambda(e)=[\tau]$ for every edge $e\in G$, is called \emph{full temporal graph} of $G$ and $\tr$. That is, every edge is available at every time step in the time span $[\tau]$. We use full temporal graphs as a technical tool.

Given a temporal path $P=\langle (e_1,t_1),(e_2,t_2),\dots,(e_k,t_k)\rangle$, we say that the \textit{departure time} of $P$ is $\td (P):=t_1$, and the \textit{arrival time} is $\ta (P):=t_k+\tr(e_k,t_k)$. Furthermore, we call the \textit{duration} of $P$ the difference between the arrival time and the departure time, i.e. $\text{dur}(P):=\big(t_k+\tr(e_k,t_k)\big)-t_1$. The \textit{traveling time} of $P$ is defined as the sum of the traversal time of its edges, that is, $\ttr (P):=\sum_{i=1}^{k} \tr(e_i,t_i)$, whereas the \textit{waiting time} of $P$ is the sum of the idle times along the path, formally $\wait(P):= \sum_{i=1}^{k-1} \big(t_{i+1}-t_i-\tr(e_i,t_i)\big)$.

Let $u$ and $v$ be two distinct vertices of $\G$ such that $u$ can temporally reach $v$. We define the following six commonly used notions of temporal distance from $u$ to $v$:
\begin{itemize}
    \item \emph{Earliest-Arrival} (\textsc{EA}) : $\EA(u,v, \G ):=\min \{\ta(P) \mid P\in\mathcal{P}_{\G}(u,v) \}$.
    \item \emph{Latest-Departure} (\textsc{LD}) : $\LD(u,v, \G ):=\max \{\td(P) \mid P\in\mathcal{P}_{\G}(u,v) \}$.
    \item \emph{Fastest-Time} (\textsc{FT}) : $\FT(u,v, \G ):=\min \{\dur(P) \mid P\in\mathcal{P}_{\G}(u,v) \}$.
    \item \emph{Shortest-Traveling} (\textsc{ST}) : $\ST(u,v, \G ):=\min \{\ttr(P) \mid P\in\mathcal{P}_{\G}(u,v) \}$.
    \item \emph{Minimum-Hop} (\textsc{MH}) : $\MH(u,v, \G ):=\min \{|P| \mid P\in\mathcal{P}_{\G}(u,v) \}$.
    \item \emph{Minimum-Waiting} (\textsc{MW}) : $\MW(u,v, \G ):=\min \{\wait (P) \mid P\in\mathcal{P}_{\G}(u,v) \}$.
\end{itemize}

For any pair of vertices, all previously defined measures can be computed in polynomial time.~\cite{branchings} (Table 1) summarize the computational complexities and the corresponding references.

Let $\D$ be any of the previous distances and let $u,v$ be two vertices. We say that $P\in \mathcal{P}_{\G}(u,v)$ \emph{realizes} $\D(u,v,\G)$ if the value of $P$ with respect to $\D$ is exactly $\D(u,v,\G)$.

An instance of \textsc{SAT} consists of a Boolean formula $\phi$ in Conjunctive Normal Form (CNF), with $q$ clauses $C_1, C_2, \dots, C_q$ over $p$ variables $\mathcal{X}=\{x_1, x_2, \dots, x_p\}$. Each clause $C_i$ is a disjunction of literals, i.e., $C_i=(l_1 \lor l_2 \lor \dots \lor l_k)$ with $k \geq 2$, where a literal is either a variable or the negation of a variable. The goal is to decide whether there exists a truth assignment $\alpha: \mathcal{X} \to \{0,1\}$ such that all clauses are satisfied under this assignment, i.e., $\phi(\alpha)=1$. When all clauses contain exactly three literals, i.e. $C_i=(l_1 \lor l_2 \lor l_3)$, the problem is called 3-\textsc{SAT}. Both \textsc{SAT} and 3-\textsc{SAT} are well-known to be \NP-complete~\cite{gareyjohnson}.

\section{\texorpdfstring{$\D$-\RF\ and\ $\D$-\TMB\ temporal optimization problems}{D-RF and D-TMB temporal optimization problems}}\label{sec:defproblems}

We now define the family of optimization problems $\D$-\RF over temporal distances $\{\EA,\LD,\FT,\ST,\MH,\MW\}$. For $\D\in \{\EA,$ $\FT,\ST,\MH,$ $\MW\}$, we define the $\D$-\RF problem as follows:

\begin{definition}[$\D$-\RF problem]
    Given a temporal graph $\G=(G,\lambda,\tr)$ and a set of sources $S\subseteq V$, the goal is to shift the temporal edges of $\G$ so that, in the resulting labeling $\lambda^*$, every $s\in S$ can temporally reach every other vertex in $(G,\lambda^*,\tr)$, and $\max_{s\in S} \max_{v \in V} \D(s,v,(G,\lambda^*,\tr))$ is minimized.
\end{definition}
For $\D=\LD$, the definition is slightly different:
\begin{definition}[$\LD$-\RF problem]
    Given a temporal graph $\G=(G,\lambda,\tr)$ and a set of sources $S\subseteq V$, the goal is to shift the temporal edges of $\G$ so that, in the resulting labeling $\lambda^*$, every $s\in S$ can temporally reach every other vertex in $(G,\lambda^*,\tr)$, and $\min_{s\in S} \min_{v \in V} \LD(s,v,(G,\lambda^*,\tr))$ is maximized.
\end{definition}

To avoid trivial negative cases, we assume that for every $s\in S$ there exists a sequence of shifts such that the resulting labeling allows $s$ to temporally reach every vertex.

Since we are allowed to shift any number of edges by an arbitrary amount, these problems can equivalently be interpreted as labeling problems. Thus, for $\D \in \{\EA,\FT,\ST,\MH,\MW\}$ we define the $\D$-\textsc{Temporal Multi-Broadcast} problem ($\D$-\TMB) as follows:
\begin{definition}[$\D$-\textsc{Temporal Multi-Broadcast} problem]
    Given a graph $G=(V,E)$, a subset of vertices $S\subseteq V$, a traversal function $\tr$, and a multiplicity function $\mu : E \to [\tau]$. Find a labeling $\lambda$ of $G$ such that for all $e\in E$, it holds $ |\lambda(e)|\leq \mu(e)$, every $s\in S$ can temporally reach every other vertex in $(G,\lambda,\tr)$, and $\max_{s\in S} \max_{v \in V} \D(s,v,(G,\lambda,\tr))$ is minimized.
\end{definition}
As in the previous case, $\LD$ requires a separate definition:
\begin{definition}[$\LD$-\textsc{Temporal Multi-Broadcast} problem]
    Given a graph $G=(V,E)$, a subset of vertices $S\subseteq V$, a traversal function $\tr$, and a multiplicity function $\mu : E \to [\tau]$. Find a labeling $\lambda$ of $G$ such that for all $e\in E$, it holds $ |\lambda(e)|\leq \mu(e)$, every $s\in S$ can temporally reach every other vertex in $(G,\lambda,\tr)$, and $\min_{s\in S} \min_{v \in V} \LD(s,v,(G,\lambda,\tr))$ is maximized.
\end{definition}
Similarly to $\D$-\RF, we avoid trivial negative cases assuming that $\tr$ is such that, for each single source $s\in S$, there exists a labeling $\lambda$ such that $s$ can temporally reach every vertex in $(G,\lambda,\tr)$.

We show that for any $\D \in \{\EA,\LD,\FT,\ST,\MH,\MW\}$ the $\D$-\RF problem and the $\D$-\TMB problem are polynomially equivalent: any instance of one problem can be transformed into an instance of the other in polynomial time, and each feasible solution for an instance of one problem can be converted into a feasible solution for the corresponding instance of the other, while preserving the value of the objective function. This highlights a close correspondence between unconstrained shifting in $\D$-\RF and assigning time labels under a multiplicity function that bounds the number of labels per edge in $\D$-\TMB.
\begin{theorem}\label{thm:equivalence}
Let $\D \in \{\EA,\LD,\FT,\ST,\MH,\MW\}$. The $\D$-\TMB and $\D$-\RF problems are equivalent.
\end{theorem}
\begin{proof}
    We first prove that $\D$-\TMB can be reduced to $\D$-\RF, and then we prove the converse. Note that the objective function is identical in both formulations.
    
    Given an instance $(G, S, \tr, \mu)$ of $\D$-\TMB, we construct a corresponding instance of $\D$-\RF as follows. We define the temporal graph $\G = (G, \lambda, \tr)$, where $\lambda(e) = [\mu(e)]$ for every $e \in E$, and $\tr$ remain unchanged. Now, given a solution $\lambda^*$ to $\D$-\RF on $(\G,S)$, it is straightforward to verify that $\lambda^*$ is also a valid solution to $\D$-\textsc{TMB}.

    Conversely, given an instance $((G,\lambda,\tr), S)$ of $\D$-\RF, we construct the multiplicity function by setting $\mu(e) = |\lambda(e)|$ for every $e \in E$. Given a solution $\lambda^*$ to $\D$-\textsc{TMB} on $(G, S,\tr,\mu)$, we need to verify that $\lambda^*$ can be obtained from $\lambda$ by shifting temporal edges.
    
    To this aim, let $e\in E$ be any edge in the underlying graph, we will prove that we can transform $\lambda(e)$ into $\lambda^*(e)$ through shifting. If $|\lambda(e)| > |\lambda^*(e)|$, we ignore the last $|\lambda(e)| - |\lambda^*(e)|$ elements of $\lambda(e)$. Moreover, we also disregard elements that appear in both sequences, since they require no shifting. We then process $\lambda(e)$ and $\lambda^*(e)$ in parallel, sorting the elements in increasing order: at each step $i$, we compare the $i$-th element of $\lambda(e)$, say $t_i$, with the $i$-th element of $\lambda^*(e)$, say $t'_i$. Finally, we observe that $x_i = t'_i - t_i$ defines a valid sequence of shifts that transforms $\lambda(e)$ into $\lambda^*(e)$.
\end{proof}

In the rest of the paper, we adopt the $\D$-\TMB formulation as it captures the same objectives as $\D$-\RF while avoiding explicit discussion of temporal edge shifts.

\section{Single source}\label{sec:singlesource}
In this section, we investigate the computational complexity of the $\D$-\TMB problems when $|S|=1$.   \cite{reachfast} proved that the \textsc{EA}-\RF problem on a single source can be solved in polynomial time, and thus so can \textsc{EA-TMB}. Alternatively, this result can be derived from \cite{spanningtreesEA} (Algorithm 2), where the authors showed that, given a temporal graph $\G$, one can compute in polynomial time a TSOT $\T_r$ rooted at any vertex $r$ that realizes all Earliest-Arrival distances, i.e., such that $\EA(r,v,\G)=\EA (r,v,\T_r)$ for all $v\in V$. Consequently, given an instance $(G,\{s\},\tr,\mu)$ of \textsc{EA-TMB}, we can construct the full temporal graph $\G_{\tr}=(G,\lambda,\tr)$ of $G$ and $\tr$.
Computing a TSOT $\T_s$ of $\G_{\tr}$ rooted at $s$ yields a feasible solution to the given instance of \textsc{EA-TMB}, as in $\T_s$ there is at most one label per edge. Moreover, $\T_s$ is an optimal solution for the \textsc{EA-TMB} instance since it guarantees that $s$ reaches \emph{every node} at the earliest possible arrival time.
\begin{theorem}\label{thm:EApolytime}[\cite{reachfast},\cite{spanningtreesEA}]
    The \textsc{EA-TMB} problem on the instance $(G,S,\tr,\mu)$ is solvable in polynomial time when $|S|=1$.
\end{theorem}

We show that \textsc{LD-TMB} can also be solved in polynomial time in the single-source setting. This turns out to be the only other measure for which the problem is tractable. We prove a stronger statement: that we can compute a TSOT $\T_r$ rooted at any vertex $r$ such that the LD distance from $r$ to any vertex is at least the minimum $\LD$ distance from $r$. 
We remark here an interesting difference with the work of ~\cite{branchings}: requiring to realize the $\LD$ distances from a root to \emph{all} the other nodes is $\NP$-hard, while requiring to realize only the worst-case $\LD$ distance is tractable.
\begin{theorem}\label{thm:LDpolytime}
    Given a temporal graph $\G=(G,\lambda,\tr)$ we can compute in polynomial-time a \textup{TSOT} $\T_r$, rooted at any vertex $r$, such that $\LD(r,v,\T_r)\geq \min_{w\in V(G)} \LD (r,w,\G)$, for all $v\in G$.
\end{theorem}
\begin{proof}
    Let $r$ be any vertex in $G$ and let $d_u$ denote the latest departure value for a temporal path from $r$ to $u$, i.e. $d_u=\LD (r,u,\G)$. Each $d_u$ is computable in polynomial time by using the algorithm from~\cite{bentert}.

    We iteratively construct a TSOT $\T_r$ rooted at $r$. Initially, the tree $\T_0$ contains only the vertex $r$. At each step $i$, we select a vertex $u_i \notin \T_{i-1}$ whose value $d_{u_i}$ is minimum among all vertices not in $\T_{i-1}$. We then consider the latest departure path $P_i=(e_1,t_1),(e_2,t_2),\dots,(e_j,t_j)$ from $r$ to $u_i$, and we construct $\T_i$ by adding the edges and vertices of $P_i$ to $\T_{i-1}$ in such a way that the resulting underlying structure remains a tree rooted at $r$. In particular, every vertex of $\T_i$ (except the root $r$) has exactly one incident edge on the unique path from $r$ to that vertex. We now show how to maintain this property.
    
    If $V(P_i)\cap V(\T_{i-1})=\emptyset$, then we set $\T_i:=\T_{i-1}\cup P_i$. Observe that in this case $r$ can temporally reach every vertex of $P_i$ with departure time $d_{u_i}$.

    Otherwise, $P_i$ contains at least one vertex already in $\T_{i-1}$. We set $\T_i=\T_{i-1}$ and process the edges of $P_i$ in order. Let $(e_k=\{w_k,w_{k+1}\},t_k)$ denote the edge currently under consideration, where $k\in [j]$.

    If $w_{k+1}\not \in \T_{i-1}$ then we add $(e_k,t_k)$ to $\T_{i}$. Observe that $d_{u_1}\leq \LD(r,w_{k+1},\T_i)\leq  d_{u_i}$. 

    Otherwise, let $(f,t_f)$ be the unique edge in $\T_{i}$ incident to $w_{k+1}$ that lies on the unique path from $r$ to $w_{k+1}$. 
    
    If $t_k+\tr(e_k,t_k)\geq t_f+\tr(e_f,t_f)$, we discard $(e_k,t_k)$ and continue with the next edge. In this case, observe that $r$ can still reach $w_{k+1}$ in time to traverse $(e_{k+1},t_{k+1})$. Moreover, the LD distance from $r$ to $w_{k+1}$ remains the same as in $\T_{i-1}$. This discard could, in principle, affect the LD distance from $r$ to $u_i$, but this can occur only when processing an edge incident to $r$. In that case, note that the LD distance from $r$ to $u_i$ is no smaller than $\min_{v\in V(\T_{i-1})} \LD (r,v,\T_{i-1})$.

    The last case occurs when $t_k+\tr(e_k,t_k)<t_f+\tr(e_f,t_f)$. In this situation, we remove $(f,t_f)$ from $\T_i$ and insert $(e_k,t_k)$ instead, which then becomes the unique edge in $\T_i$ incident to $w_{k+1}$ on the path from $r$ to $w_{k+1}$. 
    Observe that the removal of $(f,t_f)$ can affect the LD distance from $r$ to every vertex in $\T_{i-1}$ whose unique path to $r$ goes through $w_{k+1}$. However, these vertices were already in $\T_{i-1}$ and therefore had LD distance from $r$ smaller than $d_{u_i}$.
    
    Finally, this procedure ensures that $\T_i$ remains a tree rooted at $r$ at every step, with every vertex (other than $r$) having exactly one incident edge on the unique path from the root. 
    
    We repeat this process until all vertices are included in the tree and denote with $\T_r$ the resulting TSOT. Observe that at each step the LD distance from $r$ to any processed vertex does not decrease below $d_{u_1}$. Hence, $\T_r$ satisfies $\LD(r,v,\T_r)\geq \min_{w\in V(G)} \LD (r,w,\G)$ for all $v\in V(G)$.
\end{proof}
Applying Theorem~\ref{thm:LDpolytime} to the full temporal graph of $G$ and $\tr$ we obtain the following corollary.
\begin{corollary}\label{cor:LDpolytime}
    The \textsc{LD-TMB} problem on the instance $(G, S,\tr,\mu)$ is solvable in polynomial time when $|S|=1$.
\end{corollary}
\begin{proof}
    The claim follows by applying Theorem~\ref{thm:LDpolytime} to the full temporal graph $\G_{\tr}=(G,\lambda,\tr)$ of $G$ and $\tr$.
\end{proof}
On the other hand, for $\D\in\{\FT,\ST,\MH,\MW\}$, $\D$-\TMB is \NP-complete. 
\begin{theorem}\label{thm:sshard}
    Let $\D\in\{\FT,\ST,\MH,\MW\}$. It is NP-complete to decide whether $\D$-\TMB with a single source admits a feasible solution of a given value.
\end{theorem}
\begin{proof}
    We first present a reduction from \textsc{SAT} to $\FT$-\TMB, and then outline the modifications needed to establish the result for the other measures.

    Given an instance of \textsc{SAT}, $\phi = \bigwedge_{i \in [q]} C_i$ over variables $\mathcal{X}=\{x_1, x_2, \dots, x_p\}$, we construct a corresponding instance $(G,S,\tr,\mu)$ for $\FT$-\TMB as follows. 
    
    The graph $G$ contains a vertex $s$ that is the only source in $S$; for each variable $x_i$, it contains four vertices $v_{x_i}, v_{\overline{x_i}}, v_{x_i,in}, v_{x_i,out}$, and, for each clause $C_j$, it contains one vertex $v_{c_j}$.
    We add the following edges. The vertex $s$ is adjacent to all vertices of the form $v_{x_i}, v_{\overline{x_i}}$ for $i\in [p]$. The vertex $v_{x_i,in}$ has an edge to $v_{x_i}, v_{\overline{x_i}}$, and $v_{x_i,out}$. Finally, for $j\in [q]$, the vertex $v_{c_j}$ shares an edge with $v_{x_i,out}$ if $C_j$ contains the literal $x_i$ or $\overline{x_i}$. We set $\mu(e)=1$ for all edges $e$ of the form $\{v_{x_i,in},v_{x_i,out}\}$, and leave all other edges with arbitrary multiplicity.

    Let $a$ be a fixed integer with $a\geq 1$. Let $i\in [p]$ and $j\in [q]$, we construct the traversal function $\tr$ as follows. For all $x_i \in \mathcal{X}$ we set the following values for the traversal function:
\begin{itemize}
        \item for each $(e,t)\in\{(\{s,v_{x_i}\},1)$, $(\{s,v_{\overline{x_i}}\},a+1)$, $(\{v_{x_i},v_{x_{i,in}}\},2)$,\\ 
        $(\{v_{\overline{x_i}},v_{x_{i,in}}\},a +2)$,
$(\{v_{x_{i,in}},v_{x_{i,out}}\},3)$, and $(\{v_{x_{i,in}},v_{x_{i,out}}\},a + 3)\}$, $\tr(e,t)=1$;
        \item $\tr(\{v_{x_{i,out}}\,v_{c_j}\},4)=1$ if $x_i\in C_j$;
        \item $\tr(\{v_{x_{i,out}}\,v_{c_j}\},a +4)=1$ if $\overline{x_i}\in C_j$;
        \item For every other edge in $G$ and time not previously assigned, we set the traversal function $\tr$ to the value $\tau$. 
\end{itemize}

\begin{figure}[t]
\begin{center}
\resizebox{0.6\columnwidth}{!}{
\begin{tikzpicture}[]
  \tikzstyle{every node}=[inner sep=2pt,circle,black,fill=darkgray]
  
  \path (0,0) node (s){};
  \path (2,1) node (xi){};
  \path (2,-1) node (xibar){};
  \path (4,0) node (xiin){};
  \path (6,0) node (xiout){};

  \tikzstyle{every node}=[]
  \path (s) node[left=2pt,font=\Large]{$s$};
  \path (xi) node[above=2pt,font=\Large]{$v_{x_i}$};
  \path (xibar) node[below=2pt,font=\Large]{$v_{\overline{x_i}}$};
  \path (xiin) node[above=2pt,font=\Large]{$v_{x_{i,in}}$};
  \path (xiout)+(-.2,0) node[above=2pt,font=\Large]{$v_{x_{i,out}}$};
  \path (10,1) node[font=\Large]{$\{v_{c_j}:x_i \in C_j\}$};
  \path (10,-1) node[font=\Large]{$\{v_{c_j}:\overline{x_i} \in C_j\}$};

  \draw (s)-- node[above,pos=.5,font=\Large]{$1$} (xi);
  \draw (s)-- node[below=3.5pt,pos=.5,font=\Large]{$a + 1$} (xibar);
  \draw (xi)-- node[above,pos=.5,font=\Large]{$2$} (xiin);
  \draw (xibar)-- node[below=2pt,pos=.5,font=\Large]{$a +2$} (xiin);
  \draw (xiin)-- node[above,pos=.5,font=\Large]{$3$} node[below,pos=.5,font=\Large]{$a + 3$} (xiout);
  \draw[draw=black] (8,.5) rectangle ++(4,1);
  \draw[draw=black] (8,-1.5) rectangle ++(4,1);
  \draw (xiout)-- (8,.75);
  \draw (xiout)-- (8,1);
  \draw (xiout)-- node[above,pos=.3,font=\Large]{$4$} (8,1.25);
  \draw (xiout)-- (8,-.75);
  \draw (xiout)-- (8,-1);
  \draw (xiout)-- node[below=4pt,pos=.3,font=\Large]{$a + 4$} (8,-1.25);
  \draw (s)--(.5,.5);
  \draw (s)--(.5,.75);
  \draw (s)--(.5,-.5);
  \draw (s)--(.5,-.75);

\end{tikzpicture}
}
\caption{\label{fig:reduc_ss_ft}Illustration of reduction for $\D=\FT$ in Theorem~\ref{thm:sshard}. The number $x$ near an edge $e$ indicates that $\tr(e,x)=1$.}
\end{center}
\end{figure}
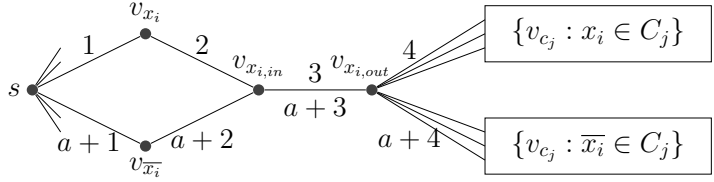
Refer to the Figure~\ref{fig:reduc_ss_ft} for an illustration of the construction. The integer $x$ near an edge $e$ indicates that $\tr(e,x)=1$. Observe that we can choose only one label for the edges of the form $\{v_{x_i,in}, v_{x_i,out}\}$, and this label corresponds to a truth value for the variable $x_i$. In particular, if we choose label 3, then $s$ can reach all clause-vertices that contain $x_i$ as a positive literal along a path of duration 4; otherwise, if we choose label $a + 3$, then $s$ can reach all clause-vertices that contain $\overline{x_i}$ as a literal along a path of duration 4.

This last observation implies that, for every $j \in [q]$, $s$ can reach $v_{c_j}$ along a path of duration 4 (using one of the two labels of the edge $\{v_{x_i,in}, v_{x_i,out}\}$). We show that the instance admits a labeling $\lambda$ with $\max_{v\in V}\FT(s,v,(G,\lambda,\tr)) = 4$ if and only if $\phi$ is satisfiable.

\begin{lemma}\label{lem:lemma2}
    \textsc{FT}-\TMB on the constructed instance $(G,\{s\},\tr, \mu)$ admits a labeling $\lambda$ such that $\max_{v\in V} \FT (s,v,(G,\lambda,\tr)) = 4$ if and only if $\phi$ is satisfiable.
\end{lemma}
\begin{proof}
We first show that if $\phi$ is satisfiable, then the instance $(G,\{s\},\tr,\mu)$ of \textsc{FT}-\TMB admits a labeling $\lambda$ with $\max_{v\in V} \FT(s,v,(G,\lambda,\tr)) = 4$.

    Let $\alpha$ be an assignment of truth values to $\mathcal{X}$ that satisfies $\phi$. Recall that we only need to choose the value of $\lambda$ for the edges of the form $\{v_{x_i,in}, v_{x_i,out}\}$ with $i\in [p]$, since including a temporal edge for which the traversal function evaluates to $\tau$ would exceed the maximum duration. If $\alpha(x_i)=1$, we set $\lambda(\{v_{x_i,in}, v_{x_i,out}\})=\{3\}$; otherwise, if $\alpha(x_i)=0$, we set $\lambda(\{v_{x_i,in}, v_{x_i,out}\})=\{a+3\}$.
        
    It is straightforward to verify that, with this $\lambda$, we have that $\max_{v\in V}\FT(s,v,(G,\lambda,\tr)) = 4$.
    In particular, let $C_j$ be any clause and let $x_i$ be the variable corresponding to a literal that evaluates to true in $C_j$ under $\alpha$. If $C_j$ contains the positive literal $x_i$, then $\lambda(\{v_{x_i,in}, v_{x_i,out}\}) = \{3\}$ and $s$ can reach $v_{c_j}$ through the temporal path $(\{s,v_{x_i}\},1),(\{v_{x_i},v_{x_i,in}\},2),(\{v_{x_i,in},v_{x_i,out}\},3)$,\\$(\{v_{x_i,out},v_{c_j}\},4)$ of duration 4. Otherwise, $C_j$ contains the literal $\overline{x_i}$, thus we have $\lambda(\{v_{x_i,in}, v_{x_i,out}\}) = \{a+3\}$ and $s$ can reach $v_{c_j}$ through the path $\{s,v_{\overline{x_i}}\},\{v_{\overline{x_i}},v_{x_i,in}\},\{v_{x_i,in},v_{x_i,out}\},\{v_{x_i,out},v_{c_j}\}$ with labels $a+1,a+2,a+3,a+4$, again of duration 4.

We now show that if \textsc{FT}-\TMB on the constructed instance $(G,\{s\},\tr, \mu)$ admits a labeling $\lambda$ such that $\max_{v\in V} \FT (s,v,(G,\lambda,\tr)) = 4$, then $\phi$ is satisfiable.

    By contradiction, assume that \textsc{FT}-\TMB on $(G,\{s\},\tr,\mu)$ admits a labeling $\lambda$ such that $\max_{v\in V}\FT (s,v,(G,\lambda,\tr)) = 4$ and that $\phi$ is not satisfiable. 

    Let $\alpha$ be the assignment of truth values to $\mathcal{X}$ corresponding to $\lambda$, that is, $\alpha(x_i) = 1$ if $\lambda(\{v_{x_i,in}, v_{x_i,out}\}) = \{3\}$, and $\alpha(x_i) = 0$ if $\lambda(\{v_{x_i,in}, v_{x_i,out}\}) = \{a+3\}$. Since $\phi$ is not satisfiable, there exists a clause $C_j$ that evaluates to false under $\alpha$. 

    Let $l$ be a literal in $C_j$ and let $x_i$ be its corresponding variable. If $\alpha(x_i)=1$, then $l$ must be a negative literal, and thus $\lambda(\{v_{x_i,out},v_{c_j}\})=\{a+4\}$ while $\lambda(\{v_{x_i,in}, v_{x_i,out}\}) = {3}$. This implies that the only path from $s$ to $v_{c_j}$ traversing the vertex $v_{x_i,in}$ is $(\{s,v_{x_i}\},1),(\{v_{x_i},v_{x_i,in}\},2)$,\\$(\{v_{x_i,in},v_{x_i,out}\},3),(\{v_{x_i,out},v_{c_j}\},a+4)$, which has duration $a+4>4$.
       
    On the other hand, if $\alpha(x_i)=0$, then $l$ must be a positive literal, and thus $\lambda(\{v_{x_i,out},v_{c_j}\})=\{4\}$ while $\lambda(\{v_{x_i,in}, v_{x_i,out}\}) = {a+3}$. Therefore, any temporal path from $s$ to $v_{c_j}$ traversing $v_{x_i,in}$ must use a temporal edge for which the traversal function is equal to $\tau$.
        
    Since we have shown that for any literal in $C_j$ the vertex $v_{c_j}$ cannot be reached from $s$ via a path of duration $4$, we obtain a contradiction to the existence of such a labeling $\lambda$.
\end{proof}

The previous lemma complete the proof for the case $\D=\FT$. Observe that Lemma~\ref{lem:lemma2} implies that the value of $\FT$-\TMB problem on an instance corresponding to an unsatisfiable SAT instance is at least $a+4$.

The constructions for $\D\in\{\ST,\MH,\MW\}$ are illustrated in Figure~\ref{fig:reductions1source}. For $\D = \ST$, the pair of integers $(x,y)$ near an edge $e$ indicates that $\tr(e,x)=y$; in the other figures, we omit $y$ since it is assumed to be~$1$.  Note that, for $\D\in\{\MH,\MW\}$, we introduce additional clause-vertices, since the traversal function of an edge is defined over all time values, allowing us to traverse an edge at any time. However, these temporal edges can only be traversed at the end of a temporal path as they have a traversal value equal to $\tau$, and the labeling function is not defined for values greater than $\tau$. 
Moreover, for $\D\in\{\ST,\MH,\MW\}$, the integer $a$ is assumed to satisfy $a\geq 2,3,1$, respectively. Additionally, for $\D=\MW$, the integer $b$ is assumed to satisfy $b\geq 2$.

As the proofs for these cases are very similar to the case of $\D=\FT$, we omit the details. We observe that the $\D$-\TMB instances corresponding to satisfiable SAT instances have value $a+3$, $a+3$, and $a$, for $\D$ equal to $\ST$, $\MH$, and $\MW$, respectively. While the instances corresponding to unsatisfiable SAT instances have value at least $2a+2$, $2a+1$, and $a(b+1)$, respectively.

\begin{figure*}[ht]
\begin{center}
\resizebox{1\textwidth}{!}{
\begin{tikzpicture}[xscale=0.6]
  \tikzstyle{every node}=[inner sep=2pt,circle,black,fill=darkgray]
  
  \path (0,0) node (s){};
  \path (-1,-2) node (xi){};
  \path (1,-2) node (xibar){};
  \path (0,-4) node (xiin){};
  \path (0,-6) node (xiout){};

  \tikzstyle{every node}=[]
  \draw (s)--(-.5,-.5);
  \draw (s)--(-.75,-.5);
  \draw (s)--(.5,-.5);
  \draw (s)--(.75,-.5);
  \path (s) node[above=2pt]{$s$};
  \path (xi) node[left=2pt]{$v_{x_i}$};
  \path (xibar) node[right=2pt]{$v_{\overline{x_i}}$};
  \path (xiin) node[right=2pt]{$v_{x_{i,in}}$};
  \path (xiout) node[right]{$v_{x_{i,out}}$};

  \draw (s)-- node[left,pos=.5]{$(1,a)$} (xi);
  \draw (s)-- node[right,pos=.5]{$(a+1,1)$} (xibar);
  \draw (xi)-- node[left,pos=.5]{$(a+1,1)$} (xiin);
  \draw (xibar)-- node[right,pos=.5]{$(a+2,1)$} (xiin);
  \draw (xiin)-- node[left,pos=.5]{$(a+2,1)$} node[right,pos=.5]{$(a+3,1)$} (xiout);
  
  \draw[draw=black] (-3.5,-8) rectangle ++(3.25,.6);
  \draw[draw=black] (0.25,-8) rectangle ++(3.25,.6);
  
  \draw (xiout)-- node[left=2pt,pos=.5]{$(a+3,1)$} (-2,-7.4);
  \draw (xiout)-- (-1.5,-7.4);
  \draw (xiout)--  (-1,-7.4);
  \draw (xiout)-- (1,-7.4);
  \draw (xiout)-- (1.5,-7.4);
  \draw (xiout)-- node[right=2pt,pos=.5]{$(a+4,a)$} (2,-7.4);
  
  \path (-1.87,-7.72) node[]{\scriptsize $\{v_{c_j}:x_i \in C_j\}$};
  \path (1.87,-7.72) node[]{\scriptsize $\{v_{c_j}:\overline{x_i} \in C_j\}$};

\end{tikzpicture}

\begin{tikzpicture}[xscale=0.6]
  \tikzstyle{every node}=[inner sep=2pt,circle,black,fill=darkgray]
  
  \path (0,0) node (s){};
  \path (-1,-1.3) node (xi1){};
  \path (-1,-2.65) node (xi2){};
  \path (1,-2) node (xibar){};
  \path (0,-4) node (xiin){};
  \path (0,-6) node (xiout){};
  \path (-2,-7.4) node (xcileft1){};
  \path (-1.5,-7.4) node (xcileft2){};
  \path (-1,-7.4) node (xcileft3){};
  \path (1,-7.4) node (xciright1){};
  \path (1.5,-7.4) node (xciright2){};
  \path (2,-7.4) node (xciright3){};
  \path (1,-8.4) node (xciright4){};
  \path (1.5,-8.4) node (xciright5){};
  \path (2,-8.4) node (xciright6){};

  \tikzstyle{every node}=[]
  \draw (s)--(-.5,-.5);
  \draw (s)--(-.75,-.5);
  \draw (s)--(.5,-.5);
  \draw (s)--(.75,-.5);
  \path (s) node[above=2pt]{$s$};
  \path (xi1) node[left=2pt]{$v_{x_{i,1}}$};
  \path (xi2) node[left=2pt]{$v_{x_{i,a-1}}$};
  \path (xibar) node[right=2pt]{$v_{\overline{x_i}}$};
  \path (xiin) node[right=2pt]{$v_{x_{i,in}}$};
  \path (xiout) node[right]{$v_{x_{i,out}}$};
  \path (-1.87,-9.72) node[]{\scriptsize $\{v_{c_j}:x_i \in C_j\}$};
  \path (1.87,-9.72) node[]{\scriptsize $\{v_{c_j}:\overline{x_i} \in C_j\}$};
  
  \draw (s)-- node[left,pos=.7]{$1$} (xi1);
  \draw (s)-- node[right,pos=.5]{$a$} (xibar);
  \draw[dashed] (-1,-1.55) -- (-1,-2.45);
  \draw (xi2)-- node[left,pos=.5]{$a$} (xiin);
  \draw (xibar)-- node[right,pos=.5]{$a+1$} (xiin);
  \draw (xiin)-- node[left,pos=.5]{$a+1$} node[right,pos=.5]{$a+2$} (xiout);
  
  \draw[draw=black] (-3.5,-10) rectangle ++(3.25,.6);
  \draw[draw=black] (0.25,-10) rectangle ++(3.25,.6);
   
  \draw (xiout)-- node[left=2pt,pos=.5]{$a+2$} (xcileft1);
  \draw (xiout)-- (xcileft2);
  \draw (xiout)--  (xcileft3);
  \draw (xcileft1)-- node[left=2pt,pos=.5]{$a+3$} (-2,-9.4);
  \draw (xcileft2)-- (-1.5,-9.4);
  \draw (xcileft3)-- (-1,-9.4);

  \draw (xiout)-- (xciright1);
  \draw (xiout)-- (xciright2);
  \draw (xiout)-- node[right=2pt,pos=.5]{$a+3$} (xciright3);
  \draw[dashed] (1,-7.6)-- (1,-8.2);
  \draw[dashed] (1.5,-7.6)-- (1.5,-8.2);
  \draw[dashed] (2,-7.6)-- (2,-8.2);
  \draw (xciright4)-- (1,-9.4) ;
  \draw (xciright5)-- (1.5,-9.4);
  \draw (xciright6)-- node[right=-2pt,pos=.5]{$2a+2$} (2,-9.4);
\end{tikzpicture}

\begin{tikzpicture}[xscale=0.6]
  \tikzstyle{every node}=[inner sep=2pt,circle,black,fill=darkgray]
  
  \path (0,0) node (s){};
  \path (-1,-2) node (xi){};
  \path (1,-2) node (xibar){};
  \path (0,-4) node (xiin){};
  \path (0,-6) node (xiout){};
  \path (-2,-7.4) node (xcileft1){};
  \path (-1.5,-7.4) node (xcileft2){};
  \path (-1,-7.4) node (xcileft3){};
  \path (1,-7.4) node (xciright1){};
  \path (1.5,-7.4) node (xciright2){};
  \path (2,-7.4) node (xciright3){};
  
  \tikzstyle{every node}=[]
  \draw (s)--(-.5,-.5);
  \draw (s)--(-.75,-.5);
  \draw (s)--(.5,-.5);
  \draw (s)--(.75,-.5);
  \path (s) node[above=2pt]{$s$};
  \path (xi) node[left=2pt]{$v_{x_i}$};
  \path (xibar) node[right=2pt]{$v_{\overline{x_i}}$};
  \path (xiin) node[right=2pt]{$v_{x_{i,in}}$};
  \path (xiout) node[right]{$v_{x_{i,out}}$};
  \path (-1.87,-9.72) node[]{\scriptsize $\{v_{c_j}:x_i \in C_j\}$};
  \path (1.87,-9.72) node[]{\scriptsize $\{v_{c_j}:\overline{x_i} \in C_j\}$};

  \draw (s)-- node[left,pos=.5]{$1$} (xi);
  \draw (s)-- node[right,pos=.5]{$ba+1$} (xibar);
  \draw (xi)-- node[left,pos=.5]{$2$} (xiin);
  \draw (xibar)-- node[right,pos=.5]{$ba+2$} (xiin);
  \draw (xiin)-- node[left,pos=.5]{$3$} node[right,pos=.5]{$ba+3$} (xiout);
  
  \draw[draw=black] (-3.5,-10) rectangle ++(3.25,.6);
  \draw[draw=black] (0.25,-10) rectangle ++(3.25,.6);
  
  \draw (xiout)-- node[left=2pt,pos=.5]{$a+4$} (xcileft1);
  \draw (xiout)-- (xcileft2);
  \draw (xiout)--  (xcileft3);
  \draw (xcileft1)-- node[left=2pt,pos=.5]{$a+5$} (-2,-9.4);
  \draw (xcileft2)-- (-1.5,-9.4);
  \draw (xcileft3)-- (-1,-9.4);
  
  \draw (xiout)-- (xciright1);
  \draw (xiout)-- (xciright2);
  \draw (xiout)-- node[right=2pt,pos=.5]{$ba+a+4$} (xciright3);
  \draw (xciright1)-- (1,-9.4) ;
  \draw (xciright2)-- (1.5,-9.4);
  \draw (xciright3)-- node[right=-3pt,pos=.5]{$ba+a+5$} (2,-9.4);

\end{tikzpicture}
}
\caption{Illustration of the necessary modifications in the reduction for $\D \in \{\ST, \MH, \MW\}$ (from left to right, respectively) in Theorem~\ref{thm:sshard}. Here, $a$ denotes an integer satisfying $a \geq 2,3,1$, respectively, and $b\geq 2$.}
\label{fig:reductions1source}
\end{center}

\end{figure*}
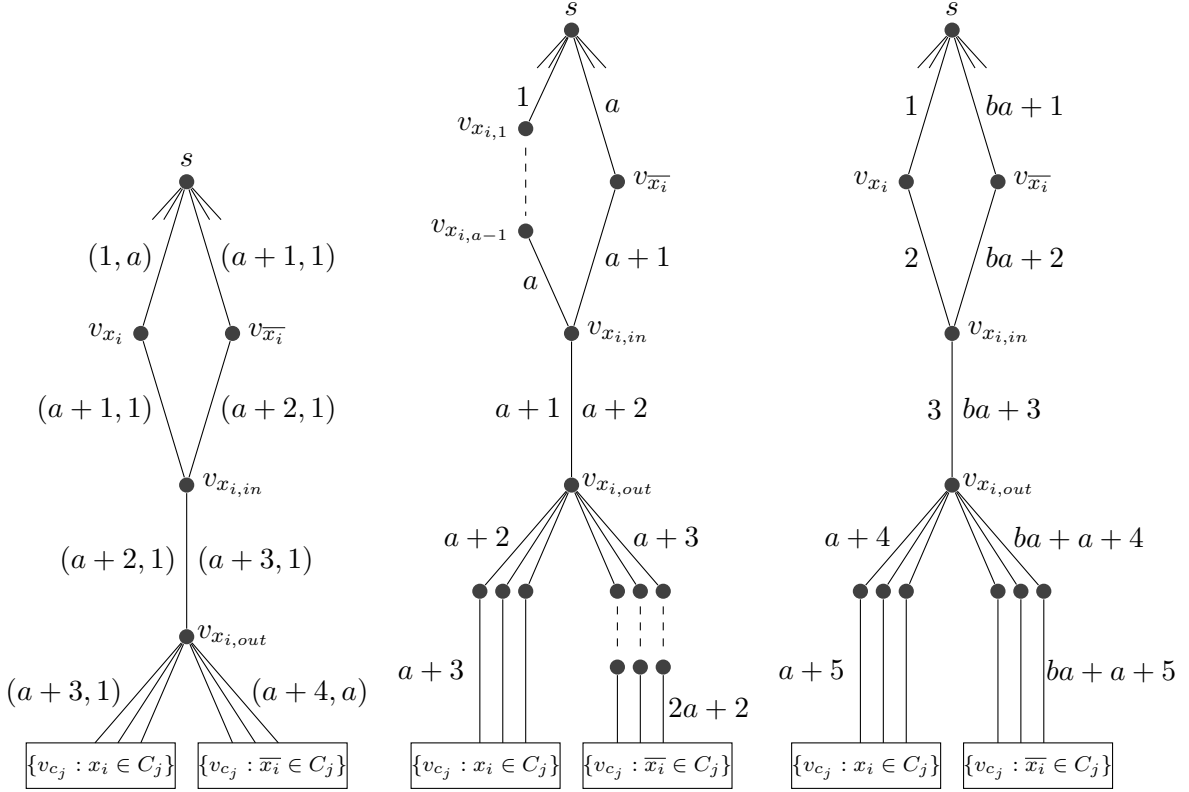

\end{proof}

The reduction in Theorem~\ref{thm:sshard} implies strong lower bounds on the approximation of $\D$-\TMB for $\D \in \{\FT,\MW\}$. To quantify such bounds, we need some further notation. For an instance $(G,\{s\},\tr,\mu)$ of $\D$-\TMB, let $\G_{\tr}=(G,\lambda,\tr)$ the full temporal graph of $G$ and $\tr$.
We define the \emph{minimum} and \emph{maximum shortest-traveling time from $s$ in $\G_{\tr}$} as $\FT_{\min}:=\max_{v\in V}\min\{\dur(P) \mid P\in\mathcal{P}(s,v)\}$ and $\FT_{\max}:=\max_{v\in V}\max\{\dur(P) \mid P\in\mathcal{P}(s,v)\}$, respectively.
Intuitively, $\FT_{\min}$ and $\FT_{\max}$ represent the minimum and maximum traveling time from $s$ to any node of $G$ if we do not consider the multiplicity constraints, respectively. In a similar way, we define the \emph{minimum} and \emph{maximum waiting time from $s$} as $\MW_{\min}:=\max_{v\in V}\min\{\wait(P) \mid P\in\mathcal{P}(s,v)\}$ and $\MW_{\max}:=\max_{v\in V}\max\{\wait(P) \mid P\in\mathcal{P}(s,v)\}$, respectively.

\begin{corollary}\label{cor:ssnoapprox}
For $\D \in \{\FT,\MW\}$, the $\D$-\TMB problem with a single source admits no approximation algorithm with approximation ratio smaller than $\FT_{\max}/\FT_{\min}$ and $\MW_{\max}/\MW_{\min}$, respectively, or smaller than $2^{|V|^{O(1)}}$, unless $\textsc{P}=\NP$.
\end{corollary}
\begin{proof}
The statements follow from a standard gap reduction argument. Suppose, by contradiction, that there exists an approximation algorithm for $\D$-\TMB with an approximation ratio $\beta$ smaller than that in the statement. Then, such an algorithm could be used to distinguish satisfiable SAT instances from unsatisfiable ones in the previous reduction, which is unlikely, unless \textsc{P} $=$ \textsc{NP}. 

To see this, note that in our reduction for $\FT$ ($\MW$, resp.), YES-instances have value at most $4$ ($a$, resp.), so the algorithm would return a solution of value at most $4\beta$ ($a\beta$, resp.), whereas NO-instances have value at least $a+4$ ($a(b+1)$, resp.).
If $\beta < \frac{a+4}{4}$ ($\beta < \frac{a(b+1)}{a}$, resp.), then satisfiable instances correspond to solutions of value smaller than $a+4$ ($a(b+1)$, resp.), while unsatisfiable ones correspond to solutions with value at least $a+4$ ($a(b+1)$, resp.). Then, the approximation algorithm could be used to distinguish satisfiable instances from unsatisfiable ones, contradicting the hardness of SAT.

The first part of the statement follows by observing that in the $\FT$ instance of the reduction $\FT_{\min} = 4$ and $\FT_{\max}=a+4$ ($\MW_{\min} = a$ and $\MW_{\max}=a(b+1)$). The second part of the statement follows by observing that $a$ ($b$, respectively) can be set to $2^{|V|^{O(1)}}$, keeping the size of the $\FT$ ($\MW$, resp.) instance polynomial in the size of the original SAT instance.
\end{proof}

Similarly, we can prove a constant lower bound on the approximation of $\D$-\TMB, when $\D \in \{\ST,\MH\}$.
\begin{corollary}\label{cor:ssnoapproxstmh}
For $\D \in \{\ST,\MH\}$, the $\D$-\TMB problem with a single source admits no approximation algorithm with a ratio better than $2$, unless $\textsc{P}=\NP$.
\end{corollary}
\begin{proof}
As in Corollary~\ref{cor:ssnoapprox}, we note that for $\ST$ ($\MH$, resp.), YES-instances of SAT have value $a+3$ ($a+3$, resp.) while NO-instances have value $2a+2$ ($2a+1$, resp.). 
If there exists an approximation algorithm for $\ST$ ($\MH$) with ratio $\beta< 2$, then by setting $a>\frac{2\beta}{2-\beta}-1$ ($a>\frac{3\beta-1}{2-\beta}$, resp.), we can distinguish satisfiable instances from unsatisfiable ones, contradicting the hardness of SAT.
\end{proof}

We conclude the section by observing that, for $\D \in \{\FT,\MW\}$, a simple algorithm guarantees an approximation factor that matches the lower bounds of Corollary~\ref{cor:ssnoapprox}.
\begin{theorem}\label{th:simplealgo}
    For $\D \in \{\FT,\MW\}$, the $\D$-\TMB problem with a single source admits an approximation algorithm with approximation ratios $\FT_{\max}/\FT_{\min}$ and $\MW_{\max}/\MW_{\min}$, respectively.
\end{theorem}
\begin{proof}
    The algorithm computes an arbitrary feasible solution to $\D$-\TMB, that is, a labeling $\lambda$ that temporally connects $s$ to every other node. It does so by using, for example, the algorithm in Theorem~\ref{thm:LDpolytime} to compute a TSOT rooted at $s$.
    By definition, the worst-case fastest-time and minimum-waiting time of this solution are at most $\FT_{\max}$ and $\MW_{\max}$, respectively. Moreover, the worst-case fastest-time and minimum-waiting time of any solution, including the optimal ones, are at least equal to $\FT_{\min}$ and $\MW_{\min}$, respectively. Therefore, the algorithm guarantees the claimed approximation ratios.
\end{proof}

\section{Multiple Sources}\label{sec:multipleSources}
Observe that in the relaxed variant where each vertex only needs to be reachable from at least one source, the multiple sources case reduces to the single source setting by adding a new source connected to all previous sources thought edges with zero traversal time at every time step in $[\tau]$; hence, the results in Section~\ref{sec:singlesource} apply. Thus, in this section we analyze the case in which every source must reach every vertex. 

In the previous section, we proved that EA-\TMB and LD-\TMB can be solved in polynomial time in the single source setting. These results extend to multiple sources in the following sense:

\begin{theorem}
    Let $\D\in\{\EA,\LD\}$. The $\D$-TMB problem on an instance $(G,S,\tr,\mu)$ with $|S| = k$ can be solved in polynomial time if, for all $e \in G$, we have $\mu(e) \geq k$.
\end{theorem}
\begin{proof}
    We prove the statement for $\D =\LD$. Let $\G_{\tr}=(G,\lambda',\tr)$ be the full temporal graph of $G$ and $\tr$.
    %
    By Theorem~\ref{thm:LDpolytime}, for each source $s\in S$ we compute a TSOT $\T_s=(G_s,\lambda_s,\tr)$, where $G_s$ is a spanning tree of $G$, and we have $\LD(s,v,\T_s)\geq \min_{w\in G} \LD (s,w,\G_{\tr})$, for all $v\in G$. Since, by construction, each edge $e\in G_s$ satisfies $|\lambda(e)|=1$, we can define a labeling $\lambda$ on $G$ by setting $\lambda(e)=\{t \mid t\in \lambda_s(e), s\in S,e\in G_s\}$. In other words, each edge of $G$ inherits the unique time-label assigned to it in any of the TSOT $\T_s$ it belongs to. Note that $|\lambda(e)|\leq k$. As we can independently maximize the minimum \textsc{LD} distance from each source to every vertex, it follows that the minimum \textsc{LD} path among all sources is also maximized.
    The proof for $\D = \EA$ is analogous, using Algorithm 1 in \cite{spanningtreesEA} instead of that in Theorem~\ref{thm:LDpolytime} used for $\LD$.
\end{proof}

On the other hand, if the multiplicity function is equal to one for every edge, then even deciding whether $\D$-TMB with two sources admits a feasible solution is \NP-complete, regardless of the objective function and for any measure $\D$. 
This implies that, without assuming the feasibility in advance, the problem is not approximable within any factor.
This resolves in the negative an open question posed in \cite{reachfast}, which asked whether a constant-factor approximation algorithm exists for \textsc{EA}-\RF.

\begin{theorem}\label{thm:noapx}
Let $\D \in \{\EA,\LD,\FT,\ST,\MH,\MW\}$. It is \NP-complete to decide whether $\D$-\TMB admits a feasible solution, even if $|S|=2$, $\mu(e)=1$, and $tr(e,t)=1$, for all $e\in E$, and $t\in [\tau]$.
\end{theorem}

We reduce 3-SAT to $\D$-\TMB. Since our proof holds for any measure $\D$, we do not specify which particular one is used in the reduction. Our construction is similar in spirit to a step in the reduction used in \cite{gobel}, where it was proved that deciding whether a graph admits a labeling with at most one label per edge such that the resulting temporal graph is temporally connected is \NP-complete. However, several non-trivial modifications are required to adapt the reduction to our problem.

An instance of \textsc{3-SAT} consists of a Conjunctive Normal Form (CNF) formula $\phi$ with $\mu$ clauses $c_1, c_2, \dots, c_{\mu}$ over $p$ variables $x_1, x_2, \dots, x_{p}$. Each clause $c_i$ consists of exactly three literals $c_i=(l_1\lor l_2 \lor l_3)$, where a literal is either a variable or the negation of a variable. We assume that no clause contains both a variable and its negation. The goal is to decide whether there exists a truth assignment to the variables such that all clauses are satisfied under this assignment. \textsc{3-SAT} is one of the 21 classical NP-complete problems identified by Karp in 1972. 

Given an instance of 3-\textsc{SAT}, $\phi = \bigwedge_{i \in [q]} C_i$ over variables $\mathcal{X}=\{x_1, x_2, \dots, x_p\}$, we construct a new formula $\phi^*$ by duplicating each variable and each clause in $\phi$. Specifically, $\phi^*$ contains the variables $x_1,x_1',x_2,x_2',\dots,x_{p},x_{p}'$, where $x_i'$ is a copy of $x_i$. For every clause $C_i = (l_1 \lor l_2 \lor l_3)$, we include $C_i$ in $\phi^*$ along with a new clause $C_i'$, which consists of the duplicated variables corresponding to $l_1$, $l_2$, and $l_3$, preserving the same polarity (i.e., positive or negated). Observe that $\phi$ is satisfiable if and only if $\phi^*$ is satisfiable. We then construct an instance $(G,S,\tr,\mu)$ of $\D$-\TMB as follows. Refer to the Figure~\ref{fig:reduc_2sources} for an illustration of the construction.

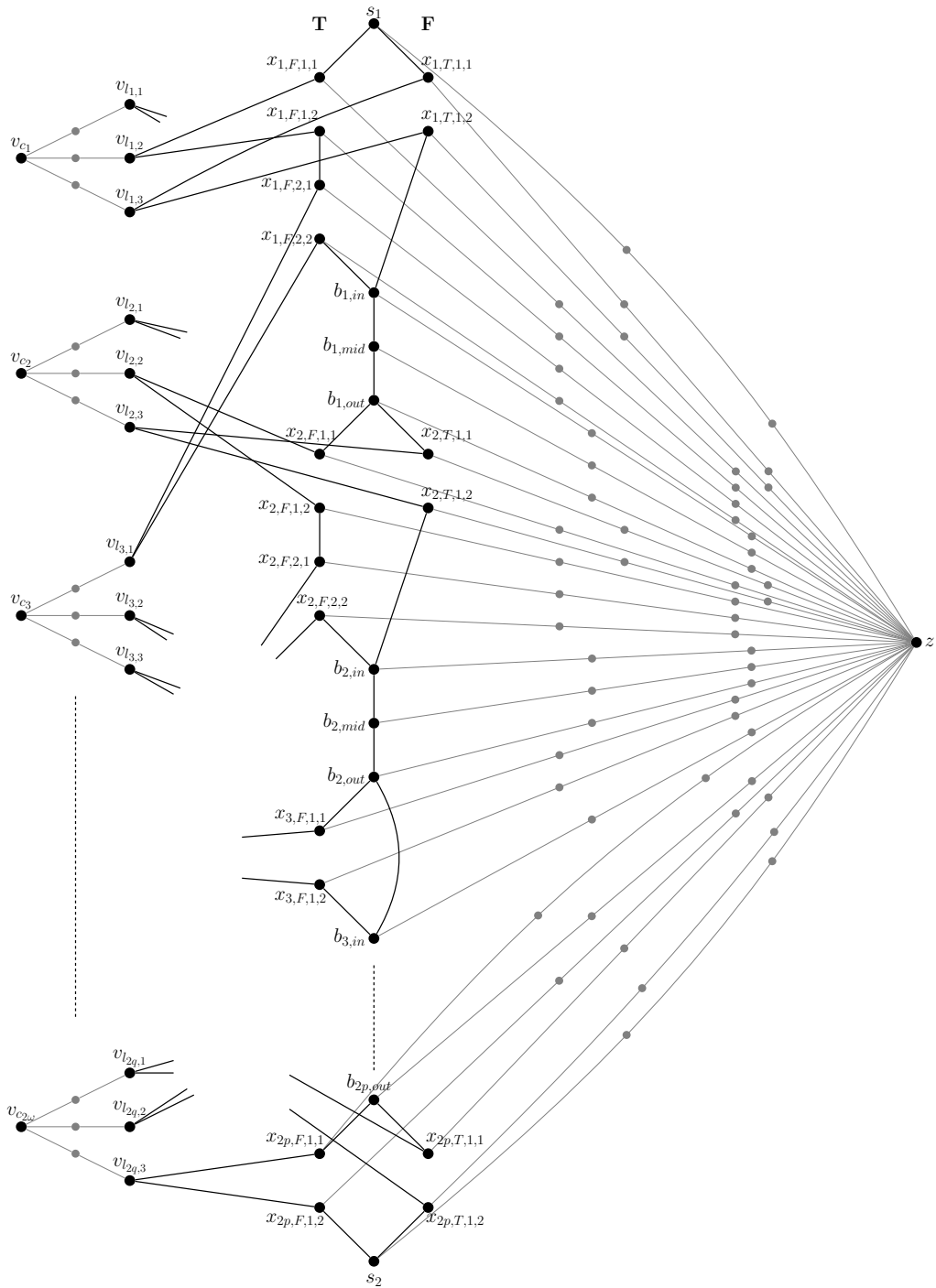
\begin{figure*}[htbp]
\centering
\resizebox{0.85\textwidth}{!}{
\begin{tikzpicture}[]
  \tikzstyle{every node}=[inner sep=4pt,circle,black,fill=black]
    \path (0,0) node (s1){};
    \path (-2,-2) node (x1t11){};
    \path (2, -2) node (x1f11){};
    \path (-2, -4) node (x1t12){};
    \path (-2, -6) node (x1t21){};
    \path (-2, -8) node (x1t22){};
    \path (2, -4) node (x1f12){};
    \path (0, -10) node (b1in){};
    \path (0, -12) node (b1mid){};
    \path (0, -14) node (b1out){};
    \path (-2, -16) node (x2t11){};
    \path (2, -16) node (x2f11){};
    \path (-2, -18) node (x2t12){};
    \path (-2, -20) node (x2t21){};
    \path (-2, -22) node (x2t22){};
    \path (2, -18) node (x2f12){};
    \path (0, -24) node (b2in){};
    \path (0, -26) node (b2mid){};
    \path (0, -28) node (b2out){};
    \path (-2, -30) node (x3t11){};
    \path (-2, -32) node (x3t12){};
    \path (0, -34) node (b3in){};
    \path (0, -40) node (b2nuout){};
    \path (-2, -42) node (x2nut11){};
    \path (-2, -44) node (x2nut12){};
    \path (2, -42) node (x2nuf11){};
    \path (2, -44) node (x2nuf12){};
    \path (0, -46) node (s2){};
    \path (-9, -3) node (l11){};
    \path (-9, -5) node (l12){};
    \path (-9, -7) node (l13){};
    \path (-13, -5) node (vc1){};
    \path (-9, -11) node (l21){};
    \path (-9, -13) node (l22){};
    \path (-9, -15) node (l23){};
    \path (-13, -13) node (vc1p){};
    \path (-9, -20) node (l31){};
    \path (-9, -22) node (l32){};
    \path (-9, -24) node (l33){};
    \path (-13, -22) node (vc2){};
    \path (20, -23) node (z){};
    \path (-9, -39) node (l2mu1){};
    \path (-9, -41) node (l2mu2){};
    \path (-9, -43) node (l2mu3){};
    \path (-13, -41) node (vc2mu){};

    \tikzstyle{every node}=[inner sep=3pt,circle,gray,fill=gray]
    \path (-11, -4) node (zl11){};
    \path (-11, -5) node (zl12){};
    \path (-11, -6) node (zl13){};
    \path (-11, -12) node (zl21){};
    \path (-11, -13) node (zl22){};
    \path (-11, -14) node (zl23){};
    \path (-11, -21) node (zl31){};
    \path (-11, -22) node (zl32){};
    \path (-11, -23) node (zl33){};
    \path (-11, -40) node (zl2mu1){};
    \path (-11, -41) node (zl2mu2){};
    \path (-11, -42) node (zl2mu3){};
    
    \draw[gray] (z)  to[bend right=10] node[pos=.3]{} node[pos=.6]{} (s1);
    \draw[gray] (z)-- node[pos=.3]{} node[pos=.6]{} (x1t11);
    \draw[gray] (z)-- node[pos=.3]{} node[pos=.6]{} (x1f11);
    \draw[gray] (z)-- node[pos=.3]{} node[pos=.6]{} (x1f12);
    \draw[gray] (z)-- node[pos=.3]{} node[pos=.6]{} (x1t12);
    \draw[gray] (z)-- node[pos=.3]{} node[pos=.6]{} (x1t21);
    \draw[gray] (z)-- node[pos=.3]{} node[pos=.6]{} (x1t22);
    \draw[gray] (z)-- node[pos=.3]{} node[pos=.6]{} (b1in);
    \draw[gray] (z)-- node[pos=.3]{} node[pos=.6]{} (b1mid);
    \draw[gray] (z)-- node[pos=.3]{} node[pos=.6]{} (b1out);
    \draw[gray] (z)-- node[pos=.3]{} node[pos=.6]{} (x2t11);
    \draw[gray] (z)-- node[pos=.3]{} node[pos=.6]{} (x2f11);
    \draw[gray] (z)-- node[pos=.3]{} node[pos=.6]{} (x2f12);
    \draw[gray] (z)-- node[pos=.3]{} node[pos=.6]{} (x2t12);
    \draw[gray] (z)-- node[pos=.3]{} node[pos=.6]{} (x2t21);
    \draw[gray] (z)-- node[pos=.3]{} node[pos=.6]{} (x2t22);
    \draw[gray] (z)-- node[pos=.3]{} node[pos=.6]{} (b2in);
    \draw[gray] (z)-- node[pos=.3]{} node[pos=.6]{} (b2mid);
    \draw[gray] (z)-- node[pos=.3]{} node[pos=.6]{} (b2out);
    \draw[gray] (z)-- node[pos=.3]{} node[pos=.6]{} (x3t11);
    \draw[gray] (z)-- node[pos=.3]{} node[pos=.6]{} (x3t12);
    \draw[gray] (z)-- node[pos=.3]{} node[pos=.6]{} (b3in);
    \draw[gray] (z)-- node[pos=.3]{} node[pos=.6]{} (b2nuout);
    \draw[gray] (z)to[bend right=10] node[pos=.3]{} node[pos=.6]{} (x2nut11);
    \draw[gray] (z)-- node[pos=.3]{} node[pos=.6]{} (x2nut12);
    \draw[gray] (z)-- node[pos=.3]{} node[pos=.6]{} (x2nuf11);
    \draw[gray] (z)to[bend left=5] node[pos=.3]{} node[pos=.6]{} (x2nuf12);
    \draw[gray] (z) to[bend left=10] node[pos=.3]{} node[pos=.6]{} (s2);
    
    \tikzstyle{every node}=[]
    \path (-2, 0) node (T){};
    \path (2, 0) node (F){};
    \path (-6.5, -39.75) node (e12){};
    \path (-6.75, -39.5) node (e11){};
    \path (-7.25, -39) node (e9){};
    \path (-7.25, -38.5) node (e10){};
    \path (-7.5, -3.5) node (e1){};
    \path (-7.75, -3.75) node (e2){};
    \path (-6.75, -11.5) node (e3){};
    \path (-7, -11.75) node (e4){};
    \path (-7.25, -22.75) node (e5){};
    \path (-7.5, -23) node (e6){};
    \path (-7, -24.75) node (e7){};
    \path (-7.25, -25) node (e8){};
    \path (-4.25, -23.25) node (e13){};
    \path (-3.75, -23.75) node (e14){};
    \path (-5, -30.25) node (e15){};
    \path (-5, -31.75) node (e16){};
    \path (-3.25, -39) node (e17){};
    \path (-3.25, -40.25) node (e18){};

    \draw[dashed, very thick] (-11, -25) to (-11, -37);
    \draw[dashed, very thick] (0, -35) to (0, -39);
    
    \path (s1) node[above=2pt]{\huge$s_1$};
    \path (s2) node[below=10pt]{\huge$s_2$};
    \path (z) node[right=5pt]{\huge$z$};
    \path (vc1) node[above=3pt]{\huge $v_{c_1}$};
    \path (vc1p) node[above=3pt]{\huge $v_{c_2}$};
    \path (vc2) node[above=3pt]{\huge $v_{c_3}$};
    \path (vc2mu) node[above=3pt]{\huge $v_{c_{2\omega}}$};
    \path (-2,0) node[]{\huge $\mathbf{T}$};
    \path (2,0) node[]{\huge $\mathbf{F}$};
    \path (l11) node[above=3pt]{\huge$v_{l_{1,1}}$};
    \path (l12) node[above=3pt]{\huge$v_{l_{1,2}}$};
    \path (l13) node[above=3pt]{\huge$v_{l_{1,3}}$};
    \path (l21) node[above=3pt]{\huge$v_{l_{2,1}}$};
    \path (l22) node[above=3pt]{\huge$v_{l_{2,2}}$};
    \path (l23) node[above=3pt]{\huge$v_{l_{2,3}}$};
    \path (l31) node[left=10pt,above=3pt]{\huge$v_{l_{3,1}}$};
    \path (l32) node[above=3pt]{\huge$v_{l_{3,2}}$};
    \path (l33) node[above=3pt]{\huge$v_{l_{3,3}}$};
    \path (l2mu1) node[above=3pt]{\huge$v_{l_{2q,1}}$};
    \path (l2mu2) node[above=6pt]{\huge$v_{l_{2q,2}}$};
    \path (l2mu3) node[above=3pt]{\huge$v_{l_{2q,3}}$};
    \path (b1in) node[left=5pt]{\huge$b_{1,in}$};
    \path (b1mid) node[left=5pt]{\huge$b_{1,mid}$};
    \path (b1out) node[left=5pt]{\huge$b_{1,out}$};
    \path (b2in) node[left=5pt]{\huge$b_{2,in}$};
    \path (b2mid) node[left=5pt]{\huge$b_{2,mid}$};
    \path (b2out) node[left=5pt]{\huge$b_{2,out}$};
    \path (b3in) node[left=5pt]{\huge$b_{3,in}$};
    \path (b2nuout) node[left=5pt,above=3pt]{\huge$b_{2p,out}$};
    \path (-3,-1.5) node[]{\huge$x_{1,F,1,1}$};
    \path (-3,-3.4) node[]{\huge$x_{1,F,1,2}$};
    \path (-3.2,-6) node[]{\huge$x_{1,F,2,1}$};
    \path (-3.2,-8) node[]{\huge$x_{1,F,2,2}$};
    \path (2.7,-1.5) node[]{\huge$x_{1,T,1,1}$};
    \path (2.7,-3.5) node[]{\huge$x_{1,T,1,2}$};
    \path (2.7,-15.3) node[]{\huge$x_{2,T,1,1}$};
    \path (2.7,-17.5) node[]{\huge$x_{2,T,1,2}$};
    \path (-2.3,-15.3) node[]{\huge$x_{2,F,1,1}$};
    \path (-3.3,-18) node[]{\huge$x_{2,F,1,2}$};
    \path (-3.3,-20) node[]{\huge$x_{2,F,2,1}$};
    \path (-1.9,-21.5) node[]{\huge$x_{2,F,2,2}$};
    \path (-2.7,-29.5) node[]{\huge$x_{3,F,1,1}$};
    \path (-2.7,-32.5) node[]{\huge$x_{3,F,1,2}$};
    
    \path (-2.9,-41.5) node[]{\huge$x_{2p,F,1,1}$};
    \path (-2.9,-44.5) node[]{\huge$x_{2p,F,1,2}$};
    \path (3,-41.5) node[]{\huge$x_{2p,T,1,1}$};
    \path (3,-44.5) node[]{\huge$x_{2p,T,1,2}$};

    \draw[very thick] (s1)-- (x1t11);
    \draw[very thick] (s1)-- (x1f11);
    \draw[very thick] (x1t11)-- (l12);
    \draw[very thick] (x1t12)-- (l12);
    \draw[very thick] (x1t12)-- (x1t21);
    \draw[very thick] (x1t21)-- (l31);
    \draw[very thick] (x1t22)-- (l31);
    \draw[very thick] (x1t22)-- (b1in);
    \draw[very thick] (x1f11) to[bend right=5] (l13);
    \draw[very thick] (x1f12)-- (l13);
    \draw[very thick] (x1f12)-- (b1in);
    \draw[very thick] (b1in)--(b1mid);
    \draw[very thick] (b1mid)--(b1out);
    \draw[very thick] (b1out)--(x2t11);
    \draw[very thick] (b1out)--(x2f11);
    \draw[very thick] (x2t11)-- (l22);
    \draw[very thick] (x2t12)-- (l22);
    \draw[very thick] (x2t12)-- (x2t21);
    \draw[very thick] (x2t21)-- (e13);
    \draw[very thick] (x2t22)-- (e14);
    \draw[very thick] (x2t22)-- (b2in);
    \draw[very thick] (x2f11)-- (l23);
    \draw[very thick] (x2f12)-- (l23);
    \draw[very thick] (x2f12)-- (b2in);
    \draw[very thick] (b2in)-- (b2mid);
    \draw[very thick] (b2mid)-- (b2out);
    \draw[very thick] (b2out)-- (x3t11);
    \draw[very thick] (b2out) to[bend left] (b3in);
    \draw[very thick] (x3t11)-- (e15);
    \draw[very thick] (x3t12)-- (e16);
    \draw[very thick] (x3t12)-- (b3in);
    \draw[very thick] (b2nuout)-- (x2nut11);
    \draw[very thick] (b2nuout)-- (x2nuf11);
    \draw[very thick] (x2nut11)-- (l2mu3);
    \draw[very thick] (x2nut12)-- (l2mu3);
    \draw[very thick] (x2nuf11)-- (e17);
    \draw[very thick] (x2nuf12)-- (e18);
    \draw[very thick] (x2nut12)-- (s2);
    \draw[very thick] (x2nuf12)-- (s2);
    \draw[very thick] (l11)-- (e1);
    \draw[very thick] (l11)-- (e2);
    \draw[very thick] (l21)-- (e3);
    \draw[very thick] (l21)-- (e4);
    \draw[very thick] (l32)-- (e5);
    \draw[very thick] (l32)-- (e6);
    \draw[very thick] (l33)-- (e7);
    \draw[very thick] (l33)-- (e8);
    \draw[very thick] (l2mu1)-- (e9);
    \draw[very thick] (l2mu1)-- (e10);
    \draw[very thick] (l2mu2)-- (e11);
    \draw[very thick] (l2mu2)-- (e12);

    \draw[gray] (vc1)-- (zl11);
    \draw[gray] (zl11)-- (l11);
    \draw[gray] (vc1)-- (zl12);
    \draw[gray] (zl12)-- (l12);
    \draw[gray] (vc1)-- (zl13);
    \draw[gray] (zl13)-- (l13);
    
    \draw[gray] (vc1p)-- (zl21);
    \draw[gray] (zl21)-- (l21);
    \draw[gray] (vc1p)-- (zl22);
    \draw[gray] (zl22)-- (l22);
    \draw[gray] (vc1p)-- (zl23);
    \draw[gray] (zl23)-- (l23);
    
    \draw[gray] (vc2)-- (zl31);
    \draw[gray] (zl31)-- (l31);
    \draw[gray] (vc2)-- (zl32);
    \draw[gray] (zl32)-- (l32);
    \draw[gray] (vc2)-- (zl33);
    \draw[gray] (zl33)-- (l33);

    \draw[gray] (vc2mu)-- (zl2mu1);
    \draw[gray] (zl2mu1)-- (l2mu1);
    \draw[gray] (vc2mu)-- (zl2mu2);
    \draw[gray] (zl2mu2)-- (l2mu2);
    \draw[gray] (vc2mu)-- (zl2mu3);
    \draw[gray] (zl2mu3)-- (l2mu3);

\end{tikzpicture}
}
\caption{Illustration of the reduction used in the proof of Theorem~\ref{thm:noapx}. Gray vertices and edges correspond to those introduced by subdivisions.}
\label{fig:reduc_2sources}
\end{figure*}

The graph $G$ includes two designated vertices, $s_1$, and $s_2$, that will act as the sources in our construction.
For each clause $C_j = (l_{j,1} \lor l_{j,2} \lor l_{j,3})$, we add a clause vertex $v_{c_j}$, three literal vertices $v_{l_{j,k}}$ for $k \in [3]$, and three auxiliary vertices $z_{l_{j,k}}$, one for each literal. We then add the edges $\{v_{c_j}, z_{l_{j,k}}\}$ and $\{z_{l_{j,k}}, v_{l_{j,k}}\}$ for each $k \in [3]$. In other words, for each clause, we construct a subdivided star centered at $v_{c_j}$ with leaves $v_{l_{j,1}}, v_{l_{j,2}}, v_{l_{j,3}}$, where each edge is subdivided by one intermediate vertex $z_{l_{j,k}}$.

For each variable $x_i$, in the order $x_1,,x_2,\dots,x_{2p}$, we construct two set of vertices $V_{x_i,T}$ and $V_{x_i,F}$ as follows. We process the clauses in the order $C_1,C_2,\dots,C_{2q}$, and for each occurrence of $x_i$ in a clause, we add two vertices, distinguishing whether the occurrence is positive or negative.

Specifically, if the $r$-th occurrence of $x_i$ as a positive literal occurs in clause $C_j$ as the $k$-th literal ($k\in [3]$), we add two vertices labeled $x_{i,F,r,1}$ and $x_{i,F,r,2}$, along with edges $\{v_{l_{j,k}},\,x_{i,F,r,1}\}$ and $\{v_{l_{j,k}},\,x_{i,F,r,2}\}$. Moreover, if $r>1$, we add the edge $\{x_{i,F,r-1,2},\,x_{i,F,r,1}\}$. We denote by $V_{x_i,F}$ the set of vertices added through this procedure.

If  the $r$-th occurrence of $x_i$ as a negative literal (i.e., $\overline{x_i}$) occurs in clause $C_j$ as the $k$-th literal ($k\in [3]$), we add two vertices labeled $x_{i,T,r,1}$ and $x_{i,T,r,2}$, along with the edges $\{v_{l_{j,k}},\,x_{i,T,r,1}\}$ and $\{v_{l_{j,k}},\,x_{i,T,r,2}\}$. Moreover, if $r>1$, we add the edge $\{x_{i,T,r-1,2},\,x_{i,T,r,1}\}$. We denote by $V_{x_i,T}$ the set of vertices added by this procedure.

For each $i \in [2p-1]$, we add the vertices $b_{i,\mathit{in}}$, $b_{i,\mathit{mid}}$, and $b_{i,\mathit{out}}$, together with the edges $\{b_{i,\mathit{in}}, b_{i,\mathit{mid}}\}$ and $\{b_{i,\mathit{mid}}, b_{i,\mathit{out}}\}$. We denote by $B$ the set of all such vertices.

Then, we add the edges $\{s_1, x_{1,F,1,1}\}$ and $\{s_1, x_{1,T,1,1}\}$ if both variable vertices exist; if only one of them exists, we add its corresponding edge and, in addition, the edge $\{s_1,b_{1,\mathit{in}}\}$. Similarly, we add the edges $\{x_{2p,F,h,2},\, s_2\}$, and $\{x_{2p,T,h',2},\, s_2\}$ if both variable vertices exist, where $h$ and $h'$ are the largest integers such that the corresponding variable vertices exist; if only one exists, we add its corresponding edge and, additionally, the edge $\{b_{2p -1,\mathit{out}},\,s_2\}$.

For each $i \in [2p - 1]$, we add the edges $\{x_{i,F,h,2},\, b_{i,\mathit{in}}\}$, $\{x_{i,T,h',2},\, b_{i,\mathit{in}}\}$, $\{b_{i,\mathit{out}},\, x_{i+1,F,1,1}\}$, and $\{b_{i,\mathit{out}},\, x_{i+1,T,1,1}\}$.
If either $h = 0$ or $h' = 0$ for some index $i \in [2p - 1]$, we add the edge $\{b_{i-1,\mathit{out}},\, b_{i,\mathit{in}}\}$.

Finally, we add a vertex $z$ that is connected to every vertex not introduced as part of a clause (i.e. $\{s_1,s_2\}\cup B \cup \bigcup_{i\in [2p]}(V_{x_i,T}\cup V_{x_i,F})$) via a path of length three, where the two internal vertices of each such path are distinct and newly created.

The constructed instance of $\D$-\TMB consists of the graph $G$, the source set $S = \{s_1, s_2\}$, multiplicity function $\mu(e) = 1$ for all $e \in E$, and traversal time function $\tr(e, t) = 1$ for all $e \in E$ and all $t \in [\tau]$.

We begin the proof of the theorem by recalling a result from~\cite{gobel}. As anticipated, our reduction introduces additional vertices compared to theirs; nevertheless, their result still applies to the graph we constructed. For completeness, we provide the proof.

\begin{lemma}\label{lemma:gobel}
    The \textsc{3-SAT} instance $\phi^*$ is satisfiable iff the graph $G$ admits a path $P$ between $s_1$ and $s_2$ such that $G \setminus E(P)$ remains connected.
\end{lemma}

\begin{proof}
    First, observe that any path from $s_1$ to $s_2$ that traverses a subdivided edge disconnects the internal vertices of that subdivided edge. Thus, we will ignore such paths in the following.

    There is a bijection between truth assignment for $\phi^*$ and paths from $s_1$ to $s_2$ in $G$. Let $\alpha$ be a truth assignment for $\phi^*$, and let $P$ be the corresponding path from $s_1$ to $s_2$ that traverses $V_{x_i,T}$ if $\alpha(x_i) = T$, and traverses $V_{x_i,F}$ if $\alpha(x_i) = F$.
    
    The only vertices that can be separated by such a path from $s_1$ to $s_2$ in $G$ are those introduced by a clause, due to the presence of the vertex $z$. Thus, given a path $P$, the graph $G \setminus E(P)$ is disconnected if and only if there exists an integer $j\in [2q]$ such that $P$ contains all three vertices $v_{l_{j,1}}$, $v_{l_{j,2}}$, and $v_{l_{j,3}}$.
    
    This occurs if and only if the three corresponding literals evaluate to false under $\alpha$, and hence the clause $C_j$ is unsatisfied. It follows that $P$ is non separating if and only if $\alpha$ is satisfiable.
\end{proof}

We are now ready to prove Theorem~\ref{thm:noapx}, which we split into two claims.

\begin{claim}\label{claim1}
    If the 3-\textsc{SAT} instance $\phi$ is satisfiable, then the constructed instance of $\D$-\TMB admits a solution.
\end{claim}
\begin{proof}
    Let $P$ be the static path between $s_1$ and $s_2$ as in Lemma~\ref{lemma:gobel}. We orient $P$ from $s_1$ to $s_2$ by transforming it into a temporal path, assigning to each edge in $P$ a time label in $\lambda$: the first edge receives time label 1, the second label 2, and so on. Let $t$ denote the label of the last edge of $P$ after this assignment.

    Since $G \setminus E(P)$ is connected, we can compute a shortest-path tree $T$ rooted at $s_2$, which is also a spanning tree. We then assign time labels to the edges of $T$ in such a way that an edge receives label $t+1$ plus its distance from $s_2$ in $T$. With this labeling, it is easy to observe that $s_1$ reaches $s_2$ at time $t + 1$, and starting from time $t + 1$, $s_2$ reaches every vertex of $G$. We thus conclude that there exists a solution to $\D$-\TMB on instance $(G,\{s_1,s_2\},\tr,\omega)$.
\end{proof}

\begin{claim}\label{claim2}
    If the \textsc{3-SAT} instance $\phi$ is not satisfiable, then the constructed instance of $\D$-\TMB admits no solution.
\end{claim}
\begin{proof}
    Assume, by contradiction, that there exists a $\lambda$ such that, in the temporal graph $(G, \lambda,\tr)$, both $s_1$ and $s_2$ can reach every vertex by some time $t$.
    
    Among all the subdivided edges incident to $z$, each consisting of a triple $(\{y_1, y_2\}, t_1)$, $(\{y_2, y_3\}, t_2)$, $(\{y_3, z\}, t_3)$, consider those satisfying $t_1 < t_2 < t_3$, where $y_1$ is either in some variable gadget $V_{x_k,T}$ or $V_{x_k,F}$, in the set $B$, or equal to $s_1$ or $s_2$.
    Such a triple must exist, since both $s_1$ and $s_2$ must be able to reach $z$. Among these, pick the triple $(\{y^*_1, y^*_2\}, t_1^*)$, $(\{y^*_2, y^*_3\}, t_2^*)$, $(\{y^*_3, z\}, t_3^*)$ for which $t_3^*$ is minimal. 
    
    Since both $s_1$ and $s_2$ must reach $y^*_2$, this is only possible if there exist two temporal paths, $Q$ from $s_1$ to $y_1^*$ and $Q'$ from $s_2$ to $y_1^*$, such that both arrive at $y_1^*$ before time $t_1^*$. Note that neither $Q$ nor $Q'$ can traverse a subdivided edge incident to $z$, as this would violate the minimality of the selected triple. We now distinguish cases based on the identity of the vertex $y_1^*$.

    \textbf{Case 1:} $y_1^*=s_1$ or $y_1^*=s_2$. Assume $y_1^*=s_1$; the other case is symmetric. In this case, by Lemma~\ref{lemma:gobel} $G\setminus E(Q')$ disconnects the literal vertices corresponding to one of the clause gadgets excluding the first one (recall that the formula has been duplicated). By the orientation induced by $Q'$ on the edges $\{b_{1,out},b_{1,mid}\}$ and $\{b_{1,mid},b_{1,in}\}$, and by the minimality of $t_3^*$, we conclude that $s_1$ cannot temporally reach the disconnected gadget, leading to a contradiction.

    \textbf{Case 2:} Assume $y_1^*\in B$. In this case, observe that $E(Q)\cup E(Q')$ forms a path from $s_1$ to $s_2$, and thus again, by Lemma~\ref{lemma:gobel} disconnects at least two clause gadgets corresponding to $C_j$ and $C_{j'}$. If one of the edge cuts is entirely contained in either $E(Q)$ or $E(Q')$ then, as before, $s_2$ or $s_1$, respectively, cannot reach the corresponding literal vertices. Thus, we may assume that both $s_1$ and $s_2$ reach at least one of the three literal vertices. Moreover, observe that the two edges incident to one of the three literal vertices are either both in $Q$ or both in $Q'$.
    
    However, we cannot label the edges in the clause gadget in such a way that both $s_1$ and $s_2$ reach all the vertices. To see this, assume that $s_1$ reaches only one literal vertex, say $v_{l_{j,1}}$, while $s_2$ reaches the other two literal vertices of a disconnected clause gadget (the other case is symmetric). Consequently, we have that $\lambda(\{ v_{l_{j,1}},z_{l_{j,1}}\}) < \lambda(\{ z_{l_{j,1}},v_{c_j}\})$, $\lambda(\{ z_{l_{j,1}},v_{c_j}\})<\lambda(\{ v_{c_j}, z_{l_{j,2}}\})$, and $\lambda(\{ z_{l_{j,1}},v_{c_j}\})<\lambda(\{ v_{c_j}, z_{l_{j,3}}\})$ as this assignment is the only one that allows $s_1$ to reach $v_{c_j}$, $v_{l_{j,2}}$, and $v_{l_{j,3}}$. However, in this case $s_2$ cannot temporally reach vertex $z_{l_{j,1}}$, yielding a contradiction.

    \textbf{Case 3:} Assume $y_1^*\in V_{x_k,T} \cup V_{x_k,F}$. The proof in this case is essentially the same as in Case 2; however, there is a small subtlety to consider. If the disconnected clause contains a literal vertex that has an edge with $y_1^*$, then both $s_1$ and $s_2$ might be able to jointly reach that literal vertex and thus all the vertices in the disconnected clause gadget. However, since we have duplicated the clauses, the other copy of the clause gadget cannot be in the same situation.

    As every case leads to a contradiction, this concludes the proof of the claim.
\end{proof}

The theorem follows from the proofs of the two preceding claims.
We note that this proof can be extended to any number of sources by adding at least four new vertices. If fewer than four sources are to be added, the same construction can be used, with some of the new vertices not acting as sources.

We add new vertices $s_3, s_4, \dots, s_{\nu}$. In addition, we insert the edges $\{s_2, s_3\}, \{s_2, s_4\}, \{s_3, s_4\}$, as well as edges of the form $\{s_j, s_3\}, \{s_j, s_4\}$ for $j = 5, 6, \dots, \nu$. Finally, we add the edge $\{s_{\nu-1}, s_{\nu}\}$.

Now, if $\phi$ is satisfiable, we can extend the labeling $\lambda$ used in Claim~\ref{claim1} with the one shown in Figure~\ref{fig:extThm7}. We only need to ensure that $\nu+1$ is the minimum label incident to $s_2$ and that $\omega$ is the maximum label incident to $s_2$. In the worst case, this can be achieved by postponing some of the labels to make $\nu+1$ the minimum label incident to $s_2$, and by adding at most $\nu$ time slots after the maximum label incident to $s_2$. Moreover, note that in the labeling shown in Figure~\ref{fig:extThm7}, any $s_i$ can reach any $s_{i'}$, with $i, i' \in {2,3,\dots,\nu}$.
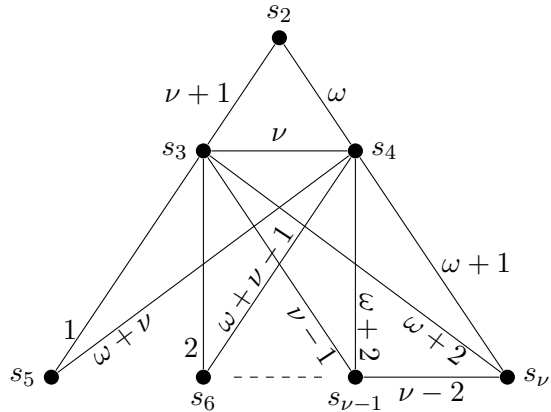
\begin{figure}[h]
\begin{center}
\begin{tikzpicture}[]
    \tikzstyle{every node}=[inner sep=2pt,circle,black,fill=black]
    \path (0,0.5) node (s2){};
    \path (-1,-1) node (s3){};
    \path (1,-1) node (s4){}; 
    \path (-3,-4) node (s5){};
    \path (-1,-4) node (s6){};
    \path (1,-4) node (s7){};
    \path (3,-4) node (s8){};
    \tikzstyle{every node}=[]
    \path (s2) node[above=2pt]{$s_2$};
    \path (s3) node[left=2pt]{$s_3$};
    \path (s4) node[right=2pt]{$s_4$};
    \path (s5) node[left=2pt]{$s_5$};
    \path (s6) node[below=2pt]{$s_6$};
    \path (s7) node[below=2pt]{$s_{\nu-1}$};
    \path (s8) node[right=2pt]{$s_{\nu}$};
    \draw (s2)-- node[left,pos=.5]{$\nu+1$} (s3);
    \draw (s2)-- node[right,pos=.5]{$\omega$} (s4);
    \draw (s3)-- node[above,pos=.5]{$\nu$} (s4);
    \draw (s3)-- node[above,pos=.9]{$1$} (s5);
    \draw (s3)-- node[left=-2pt,pos=.9]{$2$} (s6);
    \draw (s3)-- node[below=-3pt,pos=.8,sloped]{$\nu-1$} (s7);
    \draw (s3)-- node[below=-3pt,pos=.8,sloped]{$\omega+2$} (s8);
    \draw (s4)-- node[below=-3pt,pos=.8,sloped]{$\omega+\nu$} (s5);
    \draw (s4)-- node[above=-3pt,pos=.6,sloped]{$\omega+\nu-1$} (s6);
    \draw (s4)-- node[above=-3pt,pos=.8,sloped]{$\omega+2$} (s7);
    \draw (s4)-- node[right,pos=.5]{$\omega+1$} (s8);
    \draw (s7)-- node[below=-2pt,pos=.5]{$\nu-2$} (s8);
    \draw[dashed] (-.6,-4) -- (0.6,-4);
\end{tikzpicture}
\caption{Extension of the reduction in Theorem~\ref{thm:noapx} to more than two sources}
\label{fig:extThm7}
\end{center}
\end{figure}

Since all newly introduced sources can only reach vertices and be reached via $s_2$, if $\phi$ is unsatisfiable, then, by Claim~\ref{claim2}, either $s_1$ or $s_2$ cannot reach all vertices. Consequently, Claim~\ref{claim2} also holds for this extended graph.

\section{Trees}\label{sec:trees}
In this section, we present a polynomial-time algorithm for both \textsc{EA}-\TMB and \textsc{LD}-\TMB on trees with an arbitrary number of sources, under the assumption that the multiplicity function is at least two for every edge of the input graph. In fact, our algorithm requires a slightly weaker condition, that is, in each edge $e$ lying on a path between two sources, $\mu(e)\geq 2$. 
~\cite{reachfast} proposed a polynomial-time algorithm for \textsc{EA}-\TMB on trees using a different technique.
Our algorithm is given in Algorithm~\ref{alg:trees}.
\begin{algorithm}[t]
\caption{\{EA,LD\}-\TMB on Trees}\label{alg:trees}
\begin{algorithmic}
\Require $\D \in \{\EA,\LD\}$, a tree $T$, a set of sources $S\subseteq V(T)$, a traversal function $\tr$, and a multiplicity function $\mu$ such that $\mu(e)\geq 2$ for each edge $e\in T$.
\Ensure A  solution $\lambda$ to the $\D$-\TMB instance $(T,S,\tr,\mu)$.
\smallskip
\State For each $s \in S$, compute $\lambda_s$, the solution to $\D$-\TMB on instance $(T,\{s\},\tr,\mu)$ (see Section~\ref{sec:singlesource}).
\For{$e=\{u,v\} \in T$}
    \State $t_1:=0$; $t_2:=0$
    \For{$s\in S$}
        \If{$\D(s,u,(T,\lambda_s,\tr)) < \D(s,v,(T,\lambda_s,\tr))\}$}
            \State $t_1=\max(t_1,\lambda_s(\{u,v\}))$
        \Else
            \State $t_2=\max(t_2,\lambda_s(\{u,v\}))$
        \EndIf
    \EndFor
    \State $\lambda(e):=\{ t_i \mid i\in [2], t_i>0 \} $
\EndFor
\State \Return $\lambda$

\end{algorithmic}
\end{algorithm}

\begin{theorem}
    Let $\D\in \{\EA ,\LD \}$. Let $T$ be a tree, $S\subseteq V$, let $\tr$ be a traversal function, and let $\mu$ be such that $\mu(e)\geq 2$ for all $e\in E$. Algorithm~\ref{alg:trees}, on input $(T,S,\tr,\mu)$, correctly computes a solution to $\D$-\TMB on the instance $(T,S,\tr,\mu)$ in polynomial time.
\end{theorem}
\begin{proof}
    We first show that the algorithm runs in polynomial time. For each source $s$ the labeling $\lambda_s$ can be computed in polynomial time (Theorem~\ref{thm:EApolytime} and Theorem~\ref{thm:LDpolytime}). Thereafter, the algorithm only iterates over the edges of $G$ and for each edge over the single-source solutions $\lambda_s$ for $s\in S$, which can be handled in polynomial time.

    We prove correctness in two steps. First, we show that the resulting labeling preserves the property that any source can reach every other vertex. Then, we prove that the objective is minimized for \textsc{EA} and maximized for \textsc{LD}, respectively.

    Let $s$ be a source vertex, let $v \in T$ be any other vertex, let $P_{s,v} = e_1, e_2, \dots, e_k$ denote the unique (static) path from $s$ to $v$ in $T$, where $e_i=\{w_{i-1},w_i\}$ with $w_0=s$ and $w_k=v$. By assumption, in $(T, \lambda_s, \tr)$ the edges $e_1, e_2, \dots, e_k$ can be traversed in increasing order. Moreover, note that for each $e_i$ there exists $t \in \lambda(e_i)$ such that $t \geq \lambda_s(e_i)$, due to the way in which we construct $\lambda$.

    By contradiction, assume that there is no temporal path from $s$ to $v$ in $(T, \lambda, \tr)$. We may assume that $s$ and $v$ are not adjacent; otherwise, $P_{s,v} = e_1$, and since $|\lambda(e_1)| \geq 1$, we obtain a contradiction.
    
    Let $e_i=\{w_{i-1},w_i\}$ be the first edge on $P_{s,v}$ that $s$ cannot traverse in $(T, \lambda, \tr)$. Observe that $e_i\not=e_1$ as $s$ can traverse $e_1$ with any label in $\lambda(e_1)$, and $|\lambda(e_1)|\geq 1$.

    Let $t_{i-1}^*$ be the label in $\lambda$ on the edge $e_{i-1}$ at the end of Algorithm~\ref{alg:trees}, added by a source $s^*$ such that $\D(s^*,w_{i-1},(T,\lambda_{s^*},\tr))<\D(s^*,w_{i},(T,\lambda_{s^*},\tr))$. In other words, $t_{i-1}^* = \max( \{\lambda_{s'}(e_{i-1}) \mid {s'}\in S, \D(s',w_{i-2},(T,\lambda_{s'},\tr)) < \D(s',w_{i-1},(T,\lambda_{s'},\tr))\} \in \lambda_{s^*}(e_{i-1})$. Note that such a source $s^*$ must exist, as $s$ satisfies all requirements. Moreover, $s$ can traverse $(e_{i-1},t_{i-1}^*)$ by assumption.
    Then, at the end of Algorithm~\ref{alg:trees},
    $$\max( \{\lambda_{s'}(e_{i}) \mid {s'}\in S, \D(s',w_{i-2},(T,\lambda_{s'},\tr)) < \D(s',w_{i-1},(T,\lambda_{s'},\tr))\})=t_i^*\in \lambda(e_{i})$$
    and in particular $t_i^*\geq \lambda_{s^*}(e_i)$, which implies that $t_{i-1}^*<t_i^*$ and thus $s$ can traverse $(e_i,t_i^*)$, a contradiction to the assumption that there is no temporal path from $s$ to $v$ in $(T, \lambda, \tr)$. This contradiction concludes the first part of the proof.

    As a first observation for the second part of the proof, we note that every label in $\lambda$ belongs to a single-source solution $\lambda_s$, which follows directly from Algorithm~\ref{alg:trees}.
    
    Let $\D=\EA$. Let $(e_1,t_1),(e_2,t_2),\dots,(e_k,t_k)$ be one of the maximum \textsc{EA} paths from a source $s\in S$ to a vertex $v\in T$ in $(T,\lambda,\tr)$, i.e., $\EA(s,v,(T,\lambda,\tr)) = \max_{s'\in S} \max_{v' \in V} \EA(s',v',(T,\lambda,\tr))$. From the previous observation it follows that $\{t_k\}= \lambda_{s^*}(e_k)$ for some $s^*\in S$. Since $\lambda_{s^*}$ is the minimal solution for $s^*$, it follows that there is no temporal path from $s^*$ to $v$ that arrives before $t_k+\tr(e_k,t_k)$, and the same holds considering all sources together. Therefore, all the maximum \textsc{EA} paths are minimized.

    Let $\D=\LD$. Let $(e_1,t_1),(e_2,t_2),\dots,(e_k,t_k)$ be one of the minimum \textsc{LD} paths from a source $s\in S$ to a vertex $v\in T$ in $(T,\lambda,\tr)$, i.e., $\LD(s,v,(T,\lambda,\tr)) = \min_{s'\in S} \min_{v' \in V} \LD(s',v',(T,\lambda,\tr))$. From the initial observation, we have that $\{t_1\}= \lambda_{s^*}(e_1)$ for some $s^*\in S$. Since $\lambda_{s^*}$ is the maximum solution for $s^*$, then there is no temporal path from $s^*$ to $v$ with departure time greater than $t_1$, and the same holds considering all sources together. This implies that all the minimum \textsc{LD} paths are maximized.
\end{proof}

\section{Future Research Directions}\label{sec:openproblemstmb}
Building on our results, it remains open whether, for $\D\in \{\ST,\MH\}$, the $\D$-\TMB problem with a single source can be approximated within a constant factor of at least 2. Moreover, it remains an open question to identify graph classes (beyond trees for $\EA$ and $\LD$) in which the $\D$-\TMB problem is polynomially solvable, such as planar graphs or graphs of bounded treewidth.

In our formulation of the $\D$-\TMB problems, we optimize the worst-case temporal distance from the sources to all vertices. A natural variant is to optimize the average temporal distance instead. Studying the complexity and approximability of this variant is an interesting research direction.

\bibliographystyle{plainnat}
\bibliography{biblio}
\end{document}